\documentclass[12pt]{article}

\usepackage[hmargin=2.5cm, vmargin={3cm,3cm}, a4paper]{geometry}\usepackage{amsmath,amssymb,amsthm,mathrsfs}
\usepackage{cite}
\usepackage{amsthm}
\usepackage{subcaption}
\usepackage{graphicx}
\usepackage{caption}
\usepackage{url,enumerate}
\graphicspath{{fig/}}
\usepackage{float}\usepackage{parskip}
\usepackage{fancyhdr}
\usepackage{xcolor}
\usepackage{physics}
 
\DeclareMathOperator*{\argmin}{argmin} 
\DeclareMathAlphabet\mathrsfso      {U}{rsfso}{m}{n}
\usepackage{amssymb}\usepackage{mathtools}
\usepackage[toc,page]{appendix}
\usepackage{bbm}
\usepackage{centernot}

\usepackage{lettrine}
\usepackage{type1cm}
\usepackage{hyperref}
\usepackage{graphicx}
\usepackage{cleveref}

\usepackage{algorithm}
\usepackage{algpseudocode}
\usepackage{algpseudocode}

\errorcontextlines\maxdimen
\makeatletter
\newcommand*{\algrule}[1][\algorithmicindent]{\makebox[#1][l]{\hspace*{.5em}\vrule height .75\baselineskip depth .25\baselineskip}}
\newcount\ALG@printindent@tempcnta
\def\ALG@printindent{
    \ifnum \theALG@nested>0
        \ifx\ALG@text\ALG@x@notext
            \addvspace{-3pt}
        \else
            \unskip
            \ALG@printindent@tempcnta=1
            \loop
                \algrule[\csname ALG@ind@\the\ALG@printindent@tempcnta\endcsname]
                \advance \ALG@printindent@tempcnta 1
            \ifnum \ALG@printindent@tempcnta<\numexpr\theALG@nested+1\relax
            \repeat
        \fi
    \fi
    }
\usepackage{etoolbox}
\patchcmd{\ALG@doentity}{\noindent\hskip\ALG@tlm}{\ALG@printindent}{}{\errmessage{failed to patch}}
\makeatother

\newtheorem{lemma}{Lemma}

\newtheorem{proposition}{Proposition}
\newtheorem{theorem}{Theorem}

\newtheorem*{proposition*}{Proposition 3}

\newtheorem*{lemma*}{Lemma 2}
\newtheorem*{lemma**}{Lemma 1}
\newtheorem*{lemma***}{Lemma 3}
\newtheorem*{theorem*}{Theorem 1}
\newtheorem*{theorem**}{Theorem 2}

\title{Extension of Clifford Data Regression Methods for\\ Quantum Error Mitigation}
\author{Jordi Pérez-Guijarro, Alba Pagès-Zamora, and Javier R. Fonollosa}
\date{\small{SPCOM Group, Universitat Politècnica de Catalunya, Barcelona, Spain}}
\begin{document}

\maketitle
\begin{abstract}

    To address the challenge posed by noise in real quantum devices, quantum error mitigation techniques play a crucial role. These techniques are resource-efficient, making them suitable for implementation in noisy intermediate-scale quantum devices, unlike the more resource-intensive quantum error correction codes. A notable example of such a technique is Clifford Data Regression, which employs a supervised learning approach. This work investigates two variants of this technique, both of which introduce a non-trivial set of gates into the original circuit. The first variant uses multiple copies of the original circuit, while the second adds a layer of single-qubit rotations. Different characteristics of these methods are analyzed theoretically, such as their complexity, or the scaling of the error with various parameters. Additionally, the performance of these methods is evaluated through numerical experiments, demonstrating a reduction in root mean square error.

\end{abstract}
\section{Introduction}
\label{sec:intro}

The computational capacity of current quantum devices is significantly lower than that of a fault-tolerant quantum computer due to the presence of noise. For this reason, substantial research efforts have been dedicated to developing quantum error-correcting codes \cite{shor1995scheme,preskill1998reliable,kitaev1997quantum,gottesman1998theory}. However, their practical implementation on NISQ (Noisy Intermediate-Scale Quantum) devices is currently impractical due to the limited number of qubits and the high levels of noise in quantum systems.

To overcome this challenge, researchers are actively studying quantum error mitigation (QEM) techniques that can partially alleviate the effects of noise. These techniques are designed to be compatible with current quantum devices and, therefore, are not resource demanding. They require only a modest number of qubits and a small number of additional gates. This characteristic makes these techniques interesting even in the fault-tolerant regime, where they are combined with quantum error correction (QEC) to significantly reduce the required resources \cite{suzuki2022quantum,piveteau2021error}. Another distinctive aspect of QEM techniques is that they seek to estimate expectation values, i.e., expressions of the form $f(U,O)=\bra{0}^{\otimes n} U^{\dagger} O U \ket{0}^{\otimes n}$, where $U$ is a unitary matrix and $O$ is an observable, i.e., an Hermitian matrix. Note that having access to a NISQ device is essential for obtaining reasonable estimates since, under the conjecture that BQP$\neq$BPP, function $f(U,O)$ is classically hard to compute, and even hard to estimate up to $0.15  \left\| O\right\|_2$ error, where $ \left\|\cdot \right\|_2$ denotes the spectral norm \cite{huang2021power}. The conjecture BQP $\neq$ BPP is fundamental to quantum computing since, roughly speaking, it asserts that quantum computers can efficiently solve some problems that classical computers cannot.

Several error mitigation techniques have been developed, including probabilistic error cancellation (PEC) \cite{temme2017error,van2022probabilistic,tran2023locality} and zero noise extrapolation (ZNE) \cite{temme2017error,li2017efficient}. Whereas the PEC technique aims at eliminating errors by implementing the inverse of the noise channel through a sampling approach, the ZNE method estimates the expected value $f(U,O)$ by extrapolation of $f(U,O)$ estimates with different noise levels. This extrapolation method assumes a relationship between the noise level and the expected value. These two techniques face significant challenges; ZNE is very sensitive to the noise-scaling scheme and the prediction model, whereas PEC requires a precise and costly noise characterization. 

Several supervised learning techniques have been proposed to address these challenges. For instance, the work in \cite{strikis2021learning} learns the quasi probability distribution required in PEC. Similarly, the mitigation techniques in \cite{czarnik2021error,czarnik2022improving} are based on learning procedures. These techniques, known as Clifford Data Regression (CDR), approximate $f(U,O)$ by means of a parametric model that is adjusted using a training set. For these techniques to be effective, the circuits included in the training set should be efficiently simulated by a classical computer, and Clifford circuits fullfill this condition, as stated by the Gottesman-Knill theorem \cite{gottesman1998heisenberg}. The simplest parametric model considered in the literature is an affine transformation of the noisy implementation of circuit $U$ \cite{czarnik2021error,czarnik2022improving}. This is motivated by the assumption that the noise can be characterized by a depolarizing channel. This model is expanded in \cite{lowe2021unified} using a ZNE approach that incorporates additional noisy computations. Indeed, the resulting scheme can be regarded as a feature map-based model.

In this paper, we present two novel feature maps designed for CDR. The feature vectors are inspired by the same underlying principle as the ZNE method, i.e., evaluating a function at multiple positions to extrapolate an estimate of $f(U,O)$. We provide a theoretical analysis of the complexity of these methods and the scaling of the error. In addition, we conduct simulations on small circuits to evaluate the performance of these feature vectors.

The rest of the paper is organized as follows. Section \ref{sec:notation} introduces the notation and the general framework of CDR. Section \ref{sec:feature_map} presents our novel feature vectors as compared to existing approaches. Next, in Section \ref{sec:performance_analysis} the methods are analyzed theoretically. Section \ref{sec:numerical_experiments} assesses the performance of the presented methods by means of numerical experiments. Finally, conclusions are drawn in Section \ref{sec:conclusions}. 

\section{General Framework of QEM and CDR}
\label{sec:notation}

    The main goal of QEM is to accurately estimate the expected value of a specific observable, i.e., to estimate
    \begin{equation}
        f(U,O)= \bra{0}^{\otimes n} U^\dagger O U \ket{0}^{\otimes n}
    \end{equation}
    where $n$ is the number of qubits. For the remainder of the paper, we assume the observable to be fixed. Therefore, for simplicity, we denote  $f(U, O)$ as $f(U)$.
    
    Interestingly, most error mitigation schemes fall under the following general framework. First, given a classical description of a circuit $U$, $J$ perturbations of the original circuit are generated. These are denoted by
    \begin{equation}
         \{ {\mathcal{P}}_1 (U), \cdots, \mathcal{P}_J(U) \} 
    \end{equation}
    where $\mathcal{P}_j(U)$ is the $j^{th}$ perturbation of circuit $U$ implemented in an ideal device, i.e., by a unitary transformation. When the perturbed circuit is implemented in a noisy device, the resulting completely positive trace preserving (CPTP) transformation is denoted as $\tilde{\mathcal{P}}_j (U)$.

    After computing the different perturbed noisy circuits, the output state of each noisy circuit $\tilde{\mathcal{P}}_j (U)$ is measured $N$ times, and the empirical mean, denoted as $\phi(\tilde{\mathcal{P}}_j (U))$, is computed. This value is an estimate of $\Tr \left(O \tilde{\mathcal{P}}_j (U)\left(\ket{0}\bra{0}^{\otimes n}\right) \right)$, as shown by Hoeffding's inequality \cite{hoeffding1994probability}, i.e.,
    \begin{equation}\label{hoeffdings_inequality}
        \left|\phi(\tilde{\mathcal{P}}_j (U))-\Tr \left(O \tilde{\mathcal{P}}_j (U)\left(\ket{0}\bra{0}^{\otimes n}\right) \right)  \right| \leq  \|O\|_2 \sqrt{\frac{2}{N} \ln\left(\frac{2}{\delta}\right)}
    \end{equation}
    where the inequality holds with probability at least $1-\delta$. The estimate of $f(U)$ is obtained as an affine transformation of the empirical means $\phi(\tilde{\mathcal{P}}_j (U))$, 
    \begin{equation}\label{estimator_base}
        \hat{f}(U)=\sum_{j=1}^J \alpha_j \,\phi(\tilde{\mathcal{P}}_j (U))+\alpha_0  = \boldsymbol{\alpha}^T \boldsymbol{\phi}(U)
    \end{equation}
    where  $\boldsymbol{\alpha}=[\alpha_0,\alpha_1, \alpha_2, \cdots, \alpha_J]^T\in \mathbb{R}^{J+1}$, and 
    \begin{equation}\label{feature_map}
        \boldsymbol{\phi}(U)=[ \,1,\,\phi(\tilde{\mathcal{P}}_1 (U)), \phi(\tilde{\mathcal{P}}_2 (U)), \cdots, \phi(\tilde{\mathcal{P}}_J (U))]^T
    \end{equation}
    which can be regarded as a feature vector. 

    In CDR, the coefficients $\boldsymbol{\alpha}$ are not predefined. Instead, they are learned using a training set $\mathcal{T}_S=\{W_i, f(W_i)\}_{i=1}^{S}$ to better adapt to the specific noise characteristics of the device. The circuits $W_i$, which are mainly composed of Clifford gates, are designed to have a similar structure to the target circuit $U$, and, therefore, experience a similar noise when implemented on a noisy device. Since these circuits are primarily composed of Clifford gates, their expectation values $f(W_i)$ can be efficiently computed classically.
    
    Specifically, the coefficients $\boldsymbol{\alpha}$ are obtained as the solution of the unconstrained ridge regression problem, i.e.,
    \begin{align}\label{learning_alpha}
        \hat{\boldsymbol{\alpha}}&= \argmin_{\boldsymbol{\alpha} \in \mathbb{R}^{J+1}}  \left\| \boldsymbol{f}- \Phi \boldsymbol{\alpha} \right\|_2^2 + \mu \left\| \boldsymbol{\alpha} \right\|_2^2 \nonumber \\&=\left(\Phi^T \Phi + \mu I \right)^{-1} \Phi^T  \boldsymbol{f}
    \end{align}
    where $\mu \in \mathbb{R}^+$, $\boldsymbol{f}=[f(W_1),f(W_2),\cdots, f(W_S)]^T$, and $\Phi=[ \boldsymbol{\phi}(W_1), \boldsymbol{\phi}(W_2),\cdots,  \boldsymbol{\phi}(W_S)]^T$.
    
    Interestingly, since the function $f(U)$ belongs to the interval $[-\left\| O\right\|_2,\left\| O\right\|_2]$, the estimate $\hat{f}(U)$ can be slightly improved by using instead: 
        \begin{equation}\label{estimator_}
            \hat{f}(U)=\begin{cases}
 \,\,\,\left\| O\right\|_2 & \text{ if } \hat{\boldsymbol{\alpha}}^T \boldsymbol{\phi}(U)>\left\| O\right\|_2 \\
 \,\,\, \hat{\boldsymbol{\alpha}}^T \boldsymbol{\phi}(U)  & \text{ if } -\left\| O\right\|_2 \leq  \hat{\boldsymbol{\alpha}}^T \boldsymbol{\phi}(U)\leq \left\| O\right\|_2 \\
 -\left\| O\right\|_2 & \text{ if } \hat{\boldsymbol{\alpha}}^T \boldsymbol{\phi}(U)<-\left\| O\right\|_2
\end{cases}
    \end{equation}
    \subsection{Circuits of the Training set}\label{circuits_training_set}

    In this subsection, we explain the specific procedure used in this work to generate the training set circuits. First, without loss of generality, we assume that circuit $U$ is formed by the set of gates $\{\text{CNOT}, R_X, R_Y, R_Z\}$, i.e., by the CNOT gate and rotations along the $X$, $Y$, and $Z$ axes. The circuits of the training set, $W_i$, are randomly generated, while maintaining the same structure as $U$. In other words, both $U$ and $W_i$ consist of 1-qubit and 2-qubit gates applied at the same positions. 
    
    To generate these circuits $W_i$, we begin by randomly selecting a small subset of the rotations in $U$ that will remain fixed. This is done to maximize the similarity between the circuits in the training set and circuit $U$. Next, the remaining rotations are randomly replaced with Clifford gates. This step can be performed in different ways, and we choose a strategy based on the decompositions shown below \cite{bennink2017unbiased}, which allows us to study the estimation error theoretically.
    \begin{equation}\label{decomposition_X}
        \hspace{0.62cm}R_X(\theta) \rho R_X(\theta)^{\dagger} = \frac{1+\cos \theta-\sin \theta }{2} \rho +\frac{1-\cos \theta-\sin \theta }{2}  X \rho X^\dagger + \sin \theta  \, \sqrt{X} \rho \sqrt{X}^\dagger
    \end{equation}
    \begin{equation}\label{decomposition_Y}
        R_Y(\theta) \rho R_Y(\theta)^{\dagger} = \frac{1+\cos \theta-\sin \theta }{2} \rho +\frac{1-\cos \theta-\sin \theta }{2}  Y \rho Y^\dagger + \sin \theta  \, \sqrt{Y} \rho \sqrt{Y}^\dagger
    \end{equation}
    \begin{equation}\label{decomposition_Z}
        \hspace{0.65cm}R_Z(\theta) \rho R_Z(\theta)^{\dagger} = \frac{1+\cos \theta-\sin \theta }{2} \rho +\frac{1-\cos \theta-\sin \theta }{2}  Z \rho Z^\dagger + \sin \theta  \, \sqrt{Z} \rho \sqrt{Z}^\dagger
    \end{equation}

    Specifically, the non-fixed rotations are replaced by a gate from the set $\{I, X, Y, Z, \sqrt{X}, \sqrt{Y},$\\$ \sqrt{Z}\}$, with the probability of each gate being proportional to the absolute values of the corresponding coefficients in the decompositions.

    \section{Feature maps for CDR}
    \label{sec:feature_map}

    In this section, we provide an overview of the feature maps employed in previous works and introduce our proposed feature maps, i.e., the geometric feature map an the insertion feature map, highlighting the differences and motivation behind the various approaches.

    \subsection{Feature maps in prior works}
    
    A simple feature map for CDR was proposed in \cite{czarnik2021error} and it is defined as $\boldsymbol{\phi}(U)=[1,\phi(\tilde{\mathcal{U}})]^T$, where $\tilde{\mathcal{U}}$ denotes the noisy implementation of circuit $U$. This feature map is appropriate for the case in which the noise model is a depolarizing channel, i.e.,

    \begin{equation}\label{noise_expression}
       \tilde{\mathcal{U}}(\ket{0}\bra{0}^{\otimes n})= (1-p) U\ket{0}\bra{0}^{\otimes n} U^{\dagger}+p \frac{I}{2^n}
    \end{equation}
    Indeed, in this case
    \begin{equation}
        \Tr(O\tilde{\mathcal{U}}(\ket{0}\bra{0}^{\otimes n}))= (1-p) f(U)+p \frac{\Tr(O)}{2^n} 
    \end{equation}
    That is, there is an affine relation between $f(U)$ and $ \mathrm{Tr}(O\tilde{\mathcal{U}}(\ket{0}\bra{0}^{\otimes n}))$, which is approximated by $\phi(\tilde{\mathcal{U}})$ as shown in \eqref{hoeffdings_inequality}. 
    
    In most cases, this noise model in \eqref{noise_expression} is too simple, and it may be beneficial to explore more complex feature maps. For example, the ZNE-based method in \cite{lowe2021unified} uses the following feature map
    \begin{equation}
        \boldsymbol{\phi}(U)= [ \phi(\tilde{\mathcal{P}}(U,\lambda_1)),\phi(\tilde{\mathcal{P}}(U,\lambda_2)),\cdots, \phi(\tilde{\mathcal{P}}(U,\lambda_J))]
    \end{equation}
    where $\phi(\tilde{\mathcal{P}}(U,\lambda_i))$ denotes the empirical mean obtained from a circuit $U$ that has been perturbed with an artificially increased noise level. The noise level is controlled by $\lambda_i \geq 1$, a dimensionless factor such that the higher the factor, the noisier the perturbation $\tilde{\mathcal{P}}(U,\lambda_i)$ of circuit $U$.

    This feature map can be implemented in different forms depending on the sequence $\{\lambda_i\}_i$ and the method used to increase the noise. The specific implementation described in \cite{lowe2021unified} uses $\lambda_i = 2i - 1 \geq 1$, and the noise is increased using the gate folding technique \cite{giurgica2020digital} on CNOT gates. In this approach, the circuit $\tilde{\mathcal{P}}(U,\lambda_i)$ is constructed by inserting $\lambda_i - 1$ additional CNOT gates after each CNOT initially present in $U$ to amplify the noise. These additional CNOT gates are applied to the same qubits as the original ones. This noise scaling method assumes that most of the noise arises from the implementation of 2-qubit gates. Interestingly, for this type of implementation, $\lambda_i$ can be interpreted simply as the ratio between the number of CNOT gates in the perturbed circuit $\tilde{\mathcal{P}}(U,\lambda_j)$ and the original number of CNOT gates in circuit $U$.

    In the numerical experiments of this work, we implement this feature map using a slightly different approach. In particular, we use $\lambda_i=1+\frac{2(i-1)}{k}$, where $k$ denotes the number of original CNOT gates in circuit $U$. The noise is increased as in \cite{lowe2021unified} using the gate folding technique of CNOT gates \cite{giurgica2020digital, he2020resource}. Additionally, we include the constant $1$ in the feature map, as specified in \eqref{feature_map}. For more details of our particular implementation see Appendix \ref{appendix_implementation}. This appendix also includes a numerical experiment demonstrating that this implementation yields a slight improvement in the estimation error compared to \cite{lowe2021unified}.

    \subsection{Geometric Feature Map}\label{sec:Geometric_feature_map}

    The geometric feature map is defined by
    \begin{equation}\label{geometric_feature_vector_eq}
        \boldsymbol{\phi}(U)=[1,\phi(\tilde{\mathcal{U}}),\phi(\tilde{\mathcal{U}}^2),\cdots, \phi(\tilde{\mathcal{U}}^J)]^T
    \end{equation}
    where $\tilde{\mathcal{U}}^j$ denotes the composition of $\tilde{\mathcal{U}}$ $j$-times. The motivation for this extension is that the values $\{\phi(\tilde{\mathcal{U}}^j)\}_{j=2}^J$ can improve the estimate of $f(U)$. Specifically, we show that in the low noise scenario, a linear combination of these values serves as an estimator for $f(U)$. Note that this assumption is not very restrictive since QEM requires the noise level to be moderate in order to obtain reasonable results. 
    
    To understand how to approximate $f(U)$ from the values $\{\phi(\tilde{\mathcal{U}}^j)\}_{j=2}^J$, let us begin with the following definition.
    \begin{align}\label{expression_g_t_U}
        g(t,U)&:=\bra{0}^{\otimes n} (U^t)^\dagger O U^t \ket{0}^{\otimes n} \nonumber \\ & = \sum_{p=1}^{2^n} \sum_{q=1}^{2^n} \bra{0}^{\otimes n} \ket{u_q} \bra{u_q} O \ket{u_p}\bra{u_p}\ket{0}^{\otimes n} \, e^{-j 2\pi (\omega_p-\omega_q)\, t}  \nonumber \\ &=\sum_{\varphi \in \Omega(U)}  c(\varphi)\, e^{-j 2\pi \varphi t}
    \end{align}
    where $\Omega(U):=\{\omega_p-\omega_q\}_{p,q=1}^{2^n}$, $\ket{u_i}$ is the $i^{th}$ eigenvector of $U$, and $\omega_i\in[0,1)$ denotes the phase of the corresponding eigenvalue. The equality follows from substituting the spectral decomposition of $U$.

    If $\Omega(U)$ is known, the coefficients $c(\varphi)$ can be estimated by solving a linear system of equations using estimates of $g(t,U)$ for $\{t\}_{t=2}^J$. Consequently, we can also estimate $g(1,U)=f(U)$ using the estimate of these coefficients. Notably, the estimate of $f(U)$ is a linear combination of samples $\{g(t,U)\}_{t=2}^J$. If the noise is sufficiently low, then $\phi(\tilde{\mathcal{U}}^j)\approx g(j,U)$, so the estimate of $f(U)$ can be approximated as a linear combination of $\{\phi(\tilde{\mathcal{U}}^t)\}_{t=2}^J$, i.e., $f(U)\approx \alpha_2 \phi(\tilde{\mathcal{U}}^2)+\cdots + \alpha_J \phi(\tilde{\mathcal{U}}^J)$. Note, however, that the $\alpha_j$ coefficients in this linear combination are dependent on $U$ through the phase of its eigenvalues. This invalidates the learning procedure described in Section \ref{sec:notation} that calculates $\boldsymbol{\alpha}$ assuming that these coefficients remain relatively constant across the circuits in the training set.

    This conundrum can be partially overcome considering the following result.
    \begin{proposition}\label{corollary_1}
        For any unitary $U$, $p\in \mathbb{P}$ and $Q\in \mathbb{N}$, there exists a set of complex coefficients $\{c_q\}_{q=0}^{Q}$ and a constant $C$ such that
        \begin{equation}
            \left|g(t,U)-\sum_{q=-Q}^Q c_q \, e^{-j \frac{2\pi}{p}q t}\right|\leq  2\left\| O \right\|_2 \left(C \,   \frac{ \mathbbm{1}\{Q<p\}\,p}{{Q}} + \sin \left(\frac{2\pi t}{p}\right)\right)
        \end{equation}
        for $t\in[0,\frac{p}{4}]$, where $c_{-q}=c_q^*$ and $C\leq 3 \sqrt{\left( 1+ \frac{4}{3}\pi^2\right)}$.
    \end{proposition}

    \begin{proof}
        See Appendix \ref{first_appendix}      
    \end{proof}

    Therefore, depending on the desired level of precision, we can replace $\Omega(U)$ with $\{\pm i/p\}_{i=0}^Q$, which is independent of $U$. This implies that there exists a set of coefficients such that $f(U)\approx \alpha_2 \phi(\tilde{\mathcal{U}}^2)+\cdots + \alpha_J \phi(\tilde{\mathcal{U}}^J)$ for all circuits on the training set. Hence, this reasonably justifies the application of this feature vector in the context of Clifford Data Regression.

    \subsection{Insertion and Insertion-ZNE Feature Map}

    Although the geometric feature map in \eqref{geometric_feature_vector_eq} is theoretically substantiated, its practical utility is rather limited as shown in the simulations. This is mainly due to the fact that the values $\phi(\tilde{\mathcal{U}}^j)$ are extremely noisy for $j>1$. To address this issue, we propose the \textit{insertion} feature map. In this method, we split the circuit $U$ into two parts, i.e., $U=U_2U_1$, and insert a known circuit given by unitary $V$ between $U_1$ and $U_2$.
    
    Specifically, the insertion feature map is the following one
    \begin{equation}\label{insertion_feature_map}
        \boldsymbol{\phi}(U)= [ 1 , \phi(\tilde{\mathcal{U}}_2\circ \tilde{\mathcal{V}}^{\,t_1}\circ\tilde{\mathcal{U}}_1),\cdots, \phi(\tilde{\mathcal{U}}_2\circ \tilde{\mathcal{V}}^{\,t_J}\circ\tilde{\mathcal{U}}_1)]^T
    \end{equation}
    where $\tilde{\mathcal{V}}$ denotes the noisy implementation of circuit $V$. Consequently, by adopting this approach, the circuit is augmented by inserting replicas of circuit $V$, thus allowing for a controlled increase of the noise within the circuit.

    The motivation for the proposed feature map is analogous to the one presented in the previous section. In particular, let us define the following auxiliary function of $t$ to motivate this design.
    \begin{equation}
        \bra{0}^{\otimes n} U_1^\dagger (V^t)^\dagger U_2^\dagger O U_2 V^t U_1 \ket{0}^{\otimes n} 
    \end{equation}
    Note that, similarly to \eqref{expression_g_t_U}, this function can be expressed as a sum of complex exponentials whose frequencies belong to $\Omega(V)$. Therefore, following a similar procedure to the one outlined in section \ref{sec:Geometric_feature_map}, we can construct an estimator for $f(U)$ as a linear combination of values $\phi(\tilde{\mathcal{U}}_2\circ \tilde{\mathcal{V}}^{\,t_i}\circ\tilde{\mathcal{U}}_1)$. However, in this case, the coefficients in this combination are independent of the unitary $U$, since the set $\Omega(V)$ depends only on the unitary $V$. For our implementation, we use $V = R(\theta_1) \otimes \cdots \otimes R(\theta_n)$, which is composed of 1-qubit rotation gates only. This design of $V$ has the advantage of allowing us to efficiently implement the noisy circuit $\tilde{\mathcal{U}}_2 \circ \tilde{\mathcal{V}}^{\,t} \circ \tilde{\mathcal{U}}_1$ for any real value of $t$, without increasing the circuit's size. This is possible since $V^t= R(t\,\theta_1)\otimes \cdots \otimes R( t\,\theta_n)$.

    Finally, in line with \cite{lowe2021unified}, which shows that employing distinct noise levels can improve estimation, we likewise investigate the impact of introducing different levels of noise to the insertion feature map. Specifically, we also investigate the following feature map, which we refer to as the insertion-ZNE feature map
    \begin{align}\label{generalized_insetion_feature_map}
        \boldsymbol{\phi}(U)= [ 1 ,&\,\,  \phi(\tilde{\mathcal{P}}(U_2\,V^{t_1} \,U_1,\lambda_1)),\,\, \phi(\tilde{\mathcal{P}}(U_2\,V^{t_2} \,U_1,\lambda_1)),\,\cdots\,,\,\, \phi(\tilde{\mathcal{P}}(U_2\,V^{t_{J_1}} \,U_1,\lambda_1)) \nonumber\\ &  \,\, \phi(\tilde{\mathcal{P}}(U_2\,V^{t_1} \,U_1,\lambda_2)),\,\, \phi(\tilde{\mathcal{P}}(U_2\,V^{t_2} \,U_1,\lambda_2)),\,\cdots\,,\,\, \phi(\tilde{\mathcal{P}}(U_2\,V^{t_{J_1}} \,U_1,\lambda_2)) \nonumber\\& \,\,\,\,\hspace{1cm}\cdots \hspace{3cm} \cdots \hspace{3.5cm} \cdots \nonumber\\ &  \, \phi(\tilde{\mathcal{P}}(U_2\,V^{t_1} \,U_1,\lambda_{J_2})),\,\, \phi(\tilde{\mathcal{P}}(U_2\,V^{t_2} \,U_1,\lambda_{J_2})),\,\cdots\,,\phi(\tilde{\mathcal{P}}(U_2\,V^{t_{J_1}} \,U_1,\lambda_{J_2}))]^T
    \end{align}
    where $J_2$ denotes the number of noise levels, and $J_1$ is the number of different values of $t$. Therefore, the total number of perturbations is given by $J=J_1 J_2$.

    \section{Performance Analysis}\label{sec:performance_analysis}

    This section provides a theoretical analysis of several aspects of the proposed methods. In particular, we start by studying the complexity of the estimator, i.e., the number of operations required to perform the estimation. Next, we assess the estimation error through a bound derived using Rademacher complexity. Finally, an information-theoretical analysis is provided regarding the scaling of the method's variables — namely, $S$, $N$, and $J$— that is necessary to achieve specific performance levels. This latter analysis is conducted from a general standpoint and therefore applies to other types of estimators as well.

    \subsection{Complexity of the estimator}

    In this section, we examine the complexity of the resulting estimator for the feature maps presented in the previous section. The estimation method can be implemented through the following steps: generation of the training set, computation of matrix ${\Phi}$, evaluation of the optimal coefficients $\hat{\boldsymbol{\alpha}}$, and computation of estimate $\hat{f} (U)$. Next, the complexity of each step is discussed individually:

    \begin{itemize}
        \item The training set $\mathcal{T}_S=\{W_i, f(W_i)\}_{i=1}^S$ can be generated in time $O(n^2 \ell S)$, where $\ell$ denotes the number of gates of $U$. This follows since $O(n^2 \ell)$ operations are required to compute the expected value for each circuit $W_i$ \cite{aaronson2004improved}, which is mostly formed by Clifford gates. Importantly, we consider that the number of non-Clifford gates in the circuits of the training set is bounded by a constant. This aspect is significant because the scaling of the running time with the number of non-Clifford gates is exponential. In particular, for the algorithm presented in \cite{aaronson2004improved}, the running time scales as $O(16^{\tilde{\ell}})$, where $\tilde{\ell}$ denotes the number of non-Clifford gates. Interestingly, more efficient approaches have been developed to reduce the complexity when few non-Clifford gates are used \cite{bravyi2016improved,bravyi2019simulation}, although the scaling remains exponential.

        \item Next, to compute matrix $\Phi \in \mathbb{R}^{S\times (J+1)}$, which contains the feature maps of the training set circuits, the number of operations varies slightly depending on the feature map used. Since all methods can be analyzed similarly, we present the derivation only for the insertion method and, later in this section, provide the cost of the other ones. Specifically, for the insertion method, $O(S J N \ell)$ operations are required. This follows from the fact that, for each non-trivial entry of $\Phi$, we need to measure $N$ times a circuit consisting of $O(\ell)$ gates.

        \item Once $\boldsymbol{f}$ and ${\Phi}$ are computed, we can obtain $\hat{\boldsymbol{\alpha}}$ using \eqref{learning_alpha} in time $O(J^3)$. Note that there is an algorithm that can perform the inversion in time $O(J^{2.376})$ \cite{COPPERSMITH1990251}. However, the algorithm is inefficient for the values of $J$ considered because it has a quite large prefactor.

        \item Finally, the estimate $\hat{f}(U)$ can be obtained in time $O(JN\ell+J)$. The first term corresponds to the cost of computing the feature map $\boldsymbol{\phi}(U)$, and the second one to the dot product between $\hat{\boldsymbol{\alpha}}$ and $\boldsymbol{\phi}(U)$.
    \end{itemize}

    In conclusion, a total of $O(n^2 \ell S + S J N \ell +J^3)$ operations are required to compute the estimation for the insertion feature map. In comparison, when using the feature map introduced in \cite{czarnik2021error}, i.e., $\boldsymbol{\phi}(U)=[1,\phi(\tilde{\mathcal{U}})]^T$, the estimator requires a number of operations that scale as $O(n^2 \ell S + S N \ell)$. Generally, the dominant term is given by the product $S N \ell$, and therefore, roughly speaking, the ratio between both complexities is of the order $O(J)$. As shown in the numerical experiments, the proposed method reduces the average error at the cost of increasing the complexity.

    Regarding the other proposed feature maps, they require a larger number of operations than the insertion method, as shown in Table \ref{table_complexity}. This increase arises because the perturbations considered in these feature maps lead to a higher total number of gates used.

    \begin{table}[H] 
\centering 
\begin{tabular}{|c|l|}
\hline
Classical method\cite{czarnik2021error} & $O(SN\ell + n^2 \ell S)$ \\ \hline
Geometric method & $O(SNJ^2 \ell+n^2 \ell S+J^3)$ \\ \hline
Insertion method & $O(SN J \ell +n^2 \ell S +J^3)$  \\ \hline
ZNE-based method & $O(SN J \ell+SNJ^2+n^2 \ell S +J^3)$\\ \hline
Insertion-ZNE method &  $O(SNJ\ell  + SNJ J_2 +n^2 \ell S +J^3)$ \\ \hline
\end{tabular}
\caption{Complexity of the estimator for different feature maps.}
\label{table_complexity}  
\end{table}

    It is important to remark that the results shown in Table \ref{table_complexity} refer to the straightforward implementation of the method. Therefore, more efficient implementations may exist. For instance, when $J\gg S$, the kernel trick can be employed to reduce the complexity, where the kernel is defined as $\kappa(W_i,W_j):=\boldsymbol{\phi}(W_i)^T \boldsymbol{\phi}(W_j)$. This technique allows to compute $\hat{f}(U)$ without having to compute the feature maps directly; only the kernel function is needed. The estimate as a function of the kernel is given by \cite{shawe2004kernel}
    \begin{equation}\label{kernel_estimation}
        \hat{f}(U)= \boldsymbol{f}^T \left( {K} + \mu {I} \right)^{-1} \boldsymbol{k}(U)
    \end{equation}
    where $\boldsymbol{k}(U)=[\kappa(W_1,U),\cdots, \kappa(W_S,U)]^T$, and $K$ denotes the matrix formed by evaluating the kernel on the circuits of the training set, i.e., $K_{i,j}=\kappa(W_i,W_j)$. Given $\boldsymbol{f}$ and $K$, the complexity of evaluating \eqref{kernel_estimation} is $O(S^3)$, which corresponds to the cost of inverting  the kernel matrix. However, to study the complexity of the estimator, the number of operations required to approximate the kernel function $\kappa(W_i,W_j)$ should also be considered. In particular, to estimate the kernel, we first substitute \eqref{feature_map} into the definition of the kernel, yielding the following expression.
    \begin{align}\label{kernel_expression}
        \lim_{N\rightarrow \infty} \kappa(W_i,W_j)& =1 +\sum_{y=1}^J \Tr(O \tilde{\mathcal{P}}_y (W_j)(\rho_0))\Tr(O \tilde{\mathcal{P}}_y (W_i)(\rho_0)) \nonumber \\  &= 1+J\, \mathbb{E}_{Y} \Big[ \Tr((O\otimes O) (\tilde{\mathcal{P}}_Y (W_i) \otimes \tilde{\mathcal{P}}_Y (W_j)) (\rho_0^{\otimes 2}) )\Big]\nonumber  \\ & = 1+ J\,\mathbb{E}_{Y} \mathbb{E}_{A_Y}\Big[  A_Y(W_i,W_j)\Big]
    \end{align}
    where $\rho_0=\ket{0}\bra{0}^{\otimes n}$. The second equality follows from the property $\Tr(A\otimes B)=\Tr(A) \Tr(B)$, and the definition of random variable (r.v.) $Y$, which is uniformly distributed over the set $\{1,\cdots, J\}$. The final equality uses the fact that $A_y(W_i,W_j)$ denotes the r.v. obtained by measuring the state 
    \begin{equation}\label{state_form}
        \left(\tilde{\mathcal{P}}_y (W_i) \otimes \tilde{\mathcal{P}}_y (W_j)\right) (\rho_0^{\otimes 2})
    \end{equation}
    with observable $O\otimes O$. Hence, \eqref{kernel_expression} shows that to approximate the kernel, only an estimate of the expected value $\mathbb{E}_{Y} \mathbb{E}_{A_Y}[ A_Y(W_i,W_j)]$ is needed. The empirical mean, i.e., the average over $N$ samples of r.v. $A_Y(W_i,W_j)$, is used as this estimate. Since $|A_Y(W_i,W_j)|\leq \|O\|_2^2$, we can apply Hoeffding's inequality to quantify the error in this approximation. Specifically, the inequality implies that the empirical mean deviates from the expected value by at most an additive error
        \begin{equation}
            \epsilon=\|O\|_2^2\sqrt{\frac{2}{N} \ln\left(\frac{2}{\delta}\right)}
        \end{equation}
    with probability at least $1-\delta$. Thus, using the empirical mean yields and estimate of $\kappa(W_i,W_j)$ with an error $O(J\epsilon)$, which can be computed in time $O\left((\ell_i+\ell_j) N \right)$, where $\ell_i$ denotes the number of gates in circuit $W_i$. Hence, the running time of the kernel trick-based approach for the insertion method is $O\left(S^2 \ell N+n^2 \ell S+ S^3\right)$\footnote{For the other methods, the specific scaling can be derived analogously.}. The first term follows from  computing the kernel matrix $K$, the second one from the generation of the training set, and the last one from computing the estimate, which is dominated by the matrix inversion. We observe that this approach differs from the scaling of the straightforward implementation, $O(SN J \ell + n^2 \ell S + J^3)$, in two terms. Specifically, the change in the last term results from the fact that, instead of inverting a matrix of dimension $J+1$, we now invert a matrix of dimension $S$. The other change follows from the difference in computing the kernel matrix $K$ and matrix $\Phi$. In conclusion, the number of operations in the kernel trick-based approach is smaller than in the straightforward implementation only when $J \gg S$.

    \subsection{Bound on the Estimation Error}\label{error_bound_section}

    In this section, we provide an upper bound for the estimation error $|\hat{f}(U) - f(U)|$ as a function of the distinct parameters of the method. Firstly, note that even for a fixed circuit $U$, $\hat{f}(U)$ is still a r.v. due to the finite number of samples $N$ used to estimate each component of the feature map. Therefore, we study the expected value of the error,
    \begin{equation}
        \mathbb{E} \left[|\hat{f}(U)-f(U)|\right]
    \end{equation}
    which is upper bounded as (Appendix \ref{appendix_bound_error})
    \begin{equation}\label{inequality_error}
        \mathbb{E} \left[|\hat{f}(U)-f(U)|\right] \leq \frac{\|\boldsymbol{\hat{\alpha}}\|_2 \left \| O \right \|_2}{\sqrt{N}}+ |f(U)-\boldsymbol{\hat{\alpha}}^T \boldsymbol{\phi}_{\infty}(U)|
    \end{equation}
    where $\boldsymbol{\phi}_\infty(U)$ denotes the feature map obtained when $N\rightarrow \infty$ samples are used, i.e., when the expected values $\Tr \left(O \tilde{\mathcal{P}}_j (U)\left(\ket{0}\bra{0}^{\otimes n}\right) \right)$ are used as the elements of the feature map $\boldsymbol{\phi}(U)$. Now, the right-hand side of the previous inequality does not depend on the r.v. $\hat{f}(U)$; instead, it depends on $\boldsymbol{\phi}_{\infty}(U)$, which is a deterministic function.
    Next, we proceed to bound the second term. To accomplish this, we rely on two  results: firstly, the fact that any quantum channel can be decomposed into Clifford channels \cite{bennink2017unbiased}; and secondly, a generalization bound based on Rademacher complexity.

    The decomposition of a circuit into Clifford channels was used in \cite{czarnik2021error} to demonstrate that their proposed method for QEM is unbiased if zero additive error is achieved for all the Clifford circuits resulting from decomposing the desired unitary. Here, we take one step further, and provide a qualitative bound on the rate at which the error decreases.

    \begin{theorem}\label{theorem_generalization}
    Let $D(U)$ denote the distribution over circuits obtained by following the procedure explained in Section \ref{circuits_training_set}, and let $W_1, \ldots, W_S$ be $S$ i.i.d. samples from this distribution. Assuming that any 1-qubit gate is perturbed by the same quantum channel $\mathcal{N}$ when implemented on the noisy device, then, for the insertion feature map, the ZNE-based feature map, and the insertion-ZNE feature map, 
        \begin{align}\label{error_modified_}
            &\Big |f(U)-\boldsymbol{\hat{\alpha}}^T \boldsymbol{\phi}_{\infty}(U) \Big| \leq \frac{N(\boldsymbol{\theta})}{S}\sum_{i=1}^S \left| f(W_i)-\boldsymbol{\hat{\alpha}}^T \boldsymbol{\phi}_{\infty}(W_i)\right| \nonumber \\ & + \frac{\left \lceil \| \boldsymbol{\hat{\alpha}} \|_2 \right \rceil \left \| O \right \|_2  N(\boldsymbol{\theta}) \sqrt{(J+1)}}{\sqrt{S}} \left( 2+\sqrt{72 \log \frac{8 \left \lceil \| \boldsymbol{\hat{\alpha}} \|_2 \right \rceil^2}{\delta}}  \right)  
        \end{align}
        holds with probability at least $1-\delta$ over the random samples $W_i$, for any observable $O$ such that $\left \| O \right \|_2\geq 1$, where 
        \begin{equation}
            N(\boldsymbol{\theta})=\prod_{i}  \left|\frac{1+\cos \theta_i-\sin \theta_i }{2}\right|  + \left|\frac{1-\cos \theta_i-\sin \theta_i }{2} \right|  + \left|\sin \theta_i \right|  
        \end{equation}
        and the product is taken over all the rotations substituted to induce the distribution $D(U)$, with $\theta_i$ denoting the corresponding angle of the rotation.
    \end{theorem}

    \begin{proof}
        See Appendix \ref{proof_thm_1}.
    \end{proof}
    
    Analyzing the scaling of the bound, we observe that as the dimension of the feature map, $J+1$, increases, the error seems to increase. However, this picture is not completely accurate, since the error associated with the training set decreases as $J$ increases. This suggests that there should be an optimal point for the size of the feature map, which is consistent with typical results of generalization in machine learning.
    
    As a direct consequence of this theorem, if a small training error is achieved, the error $|f(U) - \boldsymbol{\hat{\alpha}}^T \boldsymbol{\phi}_{\infty}(U)|$ is primarily driven by the second term on the right-hand side of inequality \eqref{error_modified_}. That is, 
    \begin{equation}
        |f(U)-\boldsymbol{\hat{\alpha}}^T \boldsymbol{\phi}_{\infty}(U)| \approx O\left( \| \boldsymbol{\hat{\alpha}} \|_2 \left \| O \right \|_2  N(\boldsymbol{\theta}) \sqrt{\frac{ J+1}{S}} \right)
    \end{equation}
    Therefore, combining this result with \eqref{inequality_error} implies that
    \begin{align}
        \mathbb{E} \left[|\hat{f}(U)-f(U)|\right]\approx O\left(  \| \boldsymbol{\hat{\alpha}} \|_2 \left \| O \right \|_2   \left(N(\boldsymbol{\theta})\sqrt{\frac{ J+1}{S}}+\frac{1}{\sqrt{N}}\right) \right)
    \end{align}
    which shows that the expected error decays as $1/\sqrt{S}$ and $1/\sqrt{N}$. Consequently, the bias, which is upper bounded by the expected error, decreases at least at the same rate. Consequently, when $S$ and $N$ are large, the bias becomes negligible, meaning that $\mathbb{E}[\hat{f}(U)] \approx f(U)$.

    Finally, note that this theorem relies on a generalization bound that sometimes is not tight, so the true scaling could be better than the one obtained. For this reason, in Appendix \ref{numerical_experiments_testing}, we conduct numerical experiments to examine whether the error really scales as predicted by Theorem \ref{theorem_generalization}. Based on the numerical experiments, we conclude that the error does, in fact, scale as predicted by the bound.

    \subsection{Information Theoretical Lower bound }

    In the previous section, we derived an upper bound on the error, which depends on various parameters of the method, such as $S$, $J$, and $\| \boldsymbol{\hat{\alpha}} \|_2$. However, in that expression, the noise level does not appear explicitly. Intuitively, we would expect that the required computational resources to correct the noise increase as the noise level increases. To explicitly demonstrate this relationship, we adopt an information-theoretical approach similar to \cite{huang2021information}. This approach yields a lower bound on the variables $J$, $N$, and $S$.
    \begin{theorem}\label{lower_bound}
        For any set of unitaries $\{U_m\}_{m=1}^M$, such that
        \begin{equation}
            \left|f(U_i)-f(U_j) \right|\geq 2\epsilon 
        \end{equation}
        for all $i\neq j$, and for any estimator, $\hat{f}(U)$, that depends on $U$ through $\boldsymbol{\phi}(U)$ and $ \mathcal{T}_S(U)$, where
        \begin{equation}
            \boldsymbol{\phi}(U)=[ 1 ,\,  \phi(\tilde{\mathcal{P}}_1 (U)), \phi(\tilde{\mathcal{P}}_2 (U)), \cdots, \phi(\tilde{\mathcal{P}}_J (U))]^T
        \end{equation}
       and $\mathcal{T}_S(U)=\{\boldsymbol{\phi}(W_i), f(W_i)\}_{i=1}^S$,

        \begin{enumerate}[(i)]
            \item if we assume that the noise of each circuit $\tilde{\mathcal{P}_j}(U_m)$ can be modeled as a depolarizing channel with parameter $p$, i.e., $\tilde{\mathcal{P}_j}(U_m) (\rho )= (1-p){\mathcal{P}_j}(U_m)( \rho) + p \frac{I}{2^n}$, and that the same holds for the circuits in the training set, i.e., $\tilde{\mathcal{P}}_j(W) (\rho )= (1-p){\mathcal{P}}_j(W)( \rho) + p \frac{I}{2^n}$, then :

            The variables $J,S$, and $N$ satisfy 
            
                \begin{equation}\label{lower_bound_1_th}
                    \frac{(1-P_{\epsilon}-\beta(S)) \log M -H(P_{\epsilon})}{n\,\left(1- p\right)} \leq N J (S+1)
                \end{equation}

            where $H(\cdot)$ is the binary entropy function, $P_{\epsilon}$ denotes the average probability that the estimation error exceeds $\epsilon$, i.e., $\frac{1}{M} \sum_{i=1}^M \mathbb{P} (|f(U_i)-\hat{f}(U_i)|> \epsilon )$, and $\beta(S)<1$ depends on the methodology used to generate the training set (see Appendix \ref{last_appendix}).

            \item Similarly, if all circuits $\mathcal{P}_j(U_m)$ have depth $d$, and we assume that each noisy layer of gates $\tilde{\mathcal{G}}$ can be modeled as $D_p^{\otimes n}\circ \mathcal{G}$, where $D_p(\rho)$ denotes a 1-qubit depolarizing channel with parameter $p$, then
            \begin{equation}\label{lower_bound_2_th}
                \frac{(1-P_{\epsilon}-\beta(S)) \log M -H(P_{\epsilon})}{n\,\left(1- p\right)^{2d}} \leq N J (S+1)
            \end{equation}

        \end{enumerate}
    \end{theorem}

    \begin{proof}  
        See Appendix \ref{last_appendix}.
    \end{proof}

    From the bounds, it follows that at least one of $J$, $S$, or $N$ must grow at least as $O((1-p)^{-\frac{1}{3}})$ or $O((1-p)^{-\frac{2d}{3}})$, depending on the specific noise model considered. Therefore, as expected, as the noise level increases, the parameters of the model should also grow. This increase manifests as an exponential growth in the depth of the circuits. This exponential requirement is not surprising, given that other works such as \cite{takagi2023universal,quek2022exponentially,takagi2022fundamental} establish bounds consistent with Theorem \ref{lower_bound}. In particular, these works demonstrate that an exponential growth in the resources is necessary for all quantum error mitigation techniques. However, since they provide very general bounds, it remains unclear what role the different parameters of our method, namely, $N$, $S$, and $J$, or the method used to generate the training set, play. In contrast, our result shows a trade-off among the variables $N$, $S$, and $J$. Finally, note that for the insertion method, the ZNE-based method, and the insertion-ZNE method, their complexity, as shown in Table \eqref{table_complexity}, is higher than the product $SNJ$. Therefore, the lower bounds \eqref{lower_bound_1_th} and \eqref{lower_bound_2_th} directly provide lower bounds on the complexity.

    Interestingly, this theorem does not depend on the particular procedure used to process the feature map and the training set. That is, it generalizes to other models, such as random forests or multi-layer perceptron, which seem to be very competitive models \cite{liao2023machine}. Additionally, the bound is not very dependent on the approach used to generate the training set. In particular, different algorithms used to generate the training set produce distinct values of the function $\beta(S)$. The two extreme cases occur when $f(W_i)$ is independent of $U$, resulting in $\beta(S)=0$, and when $W_i$ is equal to $U$, resulting in $\beta(S)=1$, which makes the inequality trivial. Generally, we should expect the value of $\beta(S)$ to be quite small, since we are significantly modifying the circuit by substituting a large number of rotations on circuit $U$, making the values of $f(W_i)$ almost independent of $U$.

    \section{Numerical Experiments}
    \label{sec:numerical_experiments}

        \begin{figure}[!htb]
\begin{subfigure}{0.45\textwidth} 
  \includegraphics[width=\linewidth]{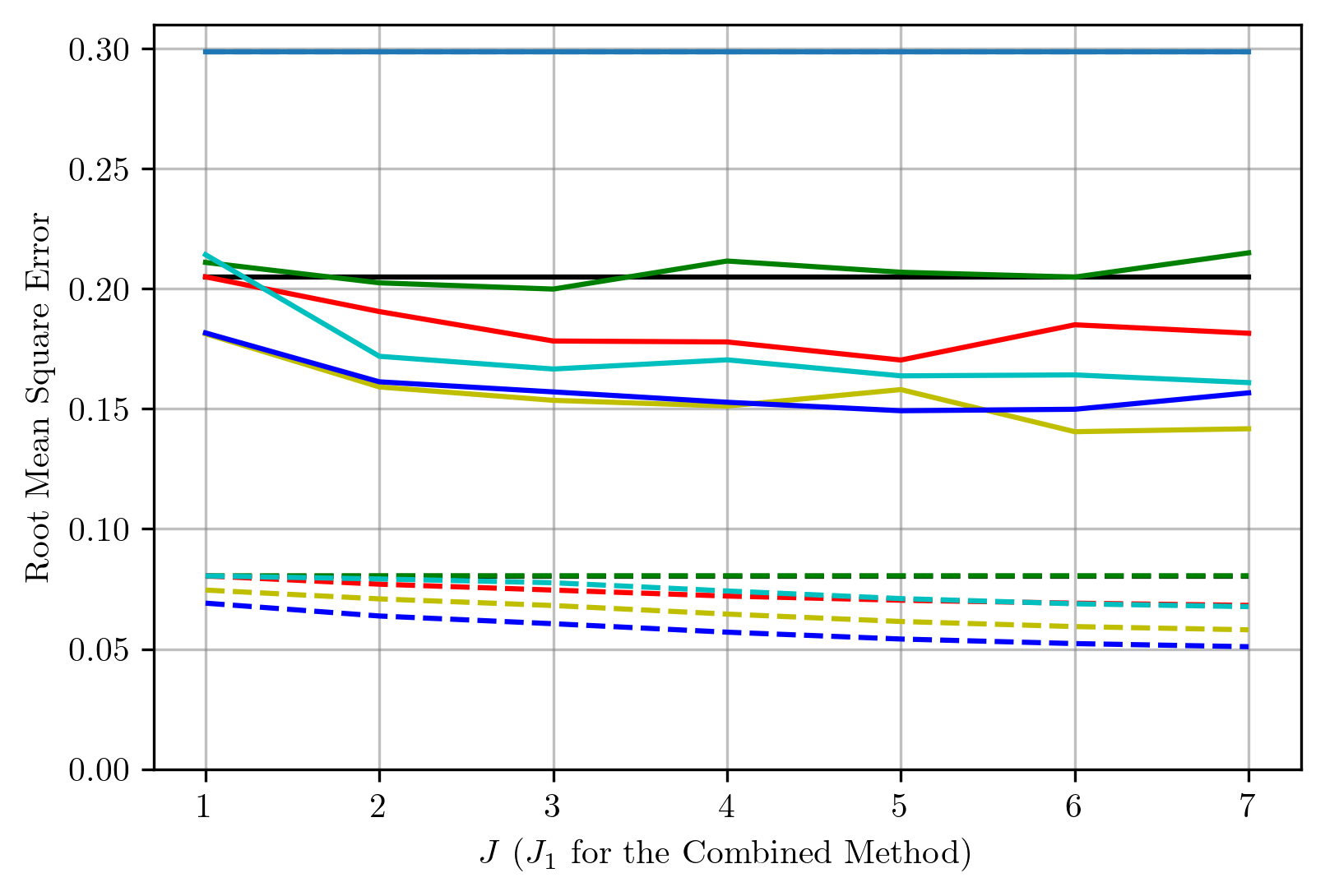}
 \caption{}
    \label{figureA}
\end{subfigure}\hfill
\begin{subfigure}{0.45\textwidth}
  \includegraphics[width=\linewidth]{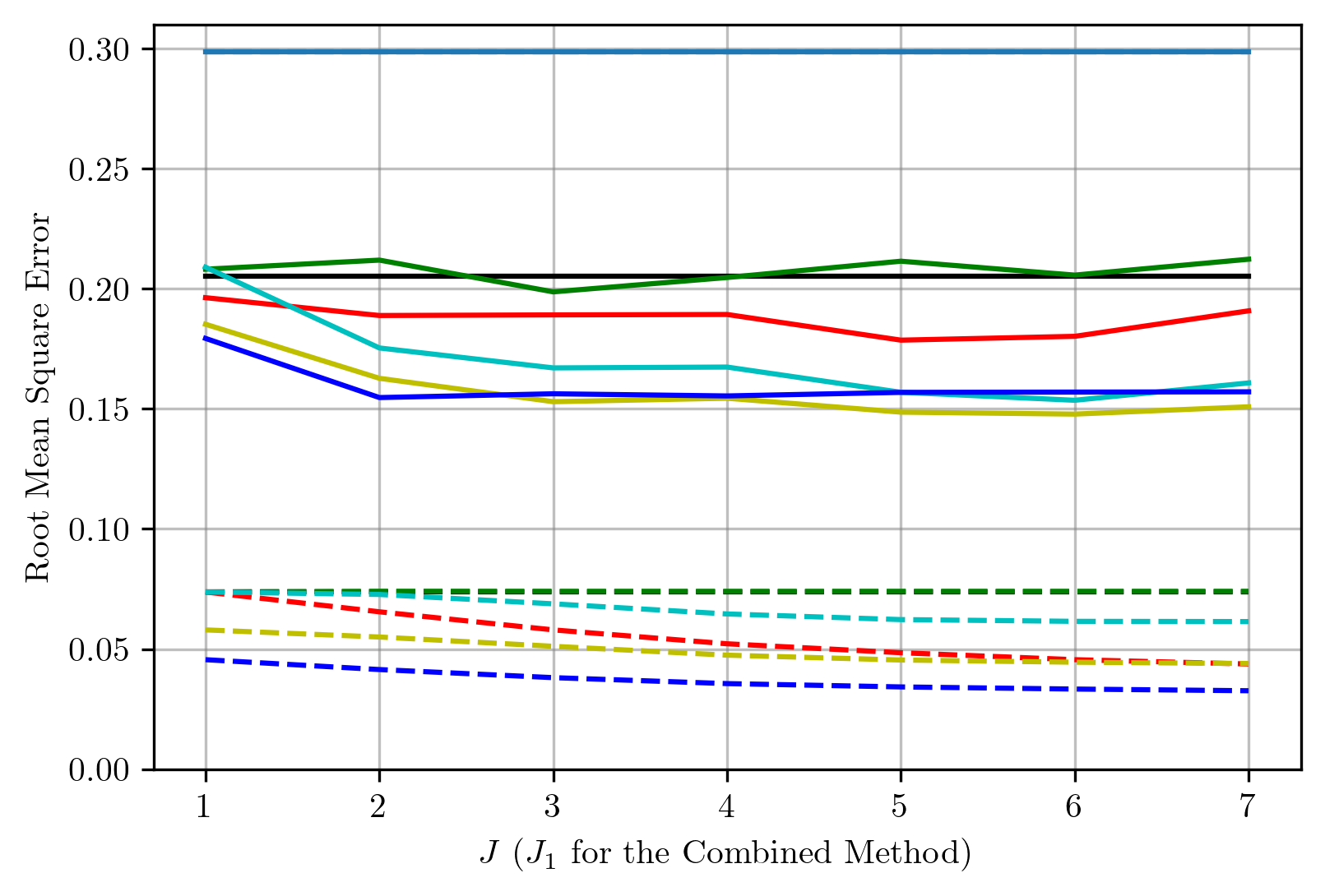}
  \caption{}
    \label{figureB}
\end{subfigure}

\begin{subfigure}{\textwidth} 
    \centering
  \includegraphics[width=0.45\linewidth]{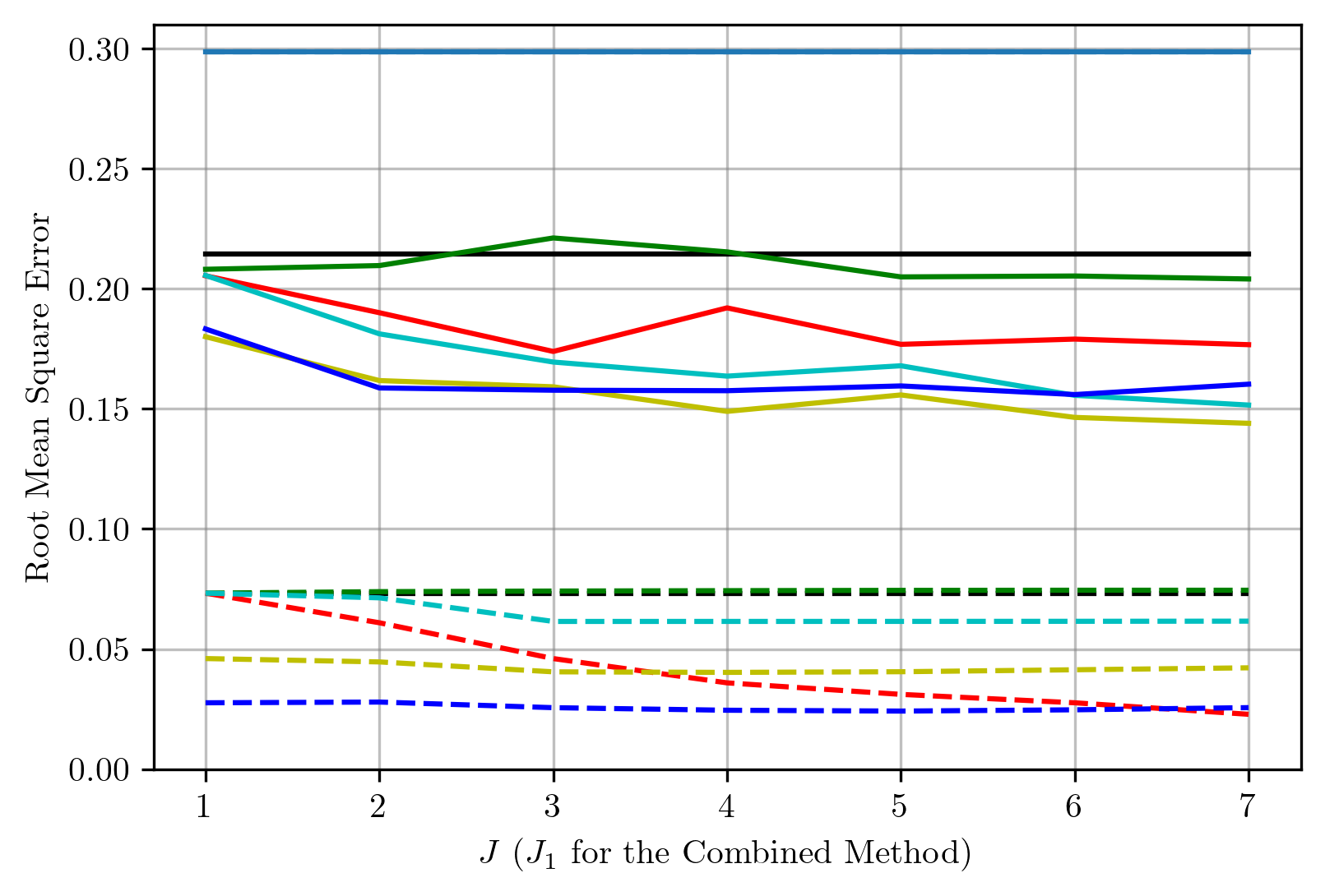}
  \caption{}
  \label{figureC} 
\end{subfigure}

\begin{subfigure}{\textwidth}
  \centering

  \includegraphics[width=0.85\linewidth]{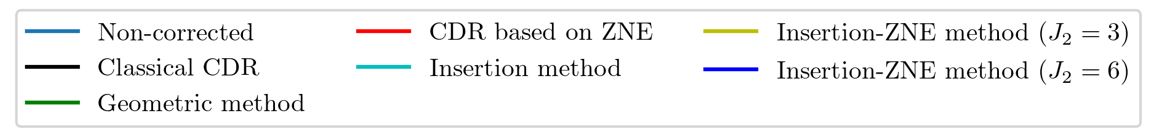}
 \vspace{0.05cm}
\end{subfigure}

\caption{Root mean square error of $\hat{f}(U)$ for different values of $J$ and various methods. Figures (a), (b), and (c) differ only in the regularization parameter used, with $\mu$ set to $10^{-2}$, $10^{-3}$, and $10^{-6}$, respectively. The solid lines represent the case where $N = 1000$ samples are used, and the dashed lines represent the ideal case where expected values are used for the computation of the feature maps.}\label{Figure_3_num}
\end{figure}

    The performance of the proposed methods is assessed by computing the root mean square error of $\hat{f}(U)$ for a set of testing circuits $\mathcal{T}_T=\{U_i\}_{i=1}^T$, i.e.,
    \begin{equation}
         \sqrt{\frac{1}{T}\sum_{i=1}^T (\hat{f} (U_i)-f(U_i))^2}
    \end{equation}
    where $\hat{f}(U)$ is given by \eqref{estimator_}, with the feature map varying depending on the method. 

    All testing circuits $\{U_i\}_{i=1}^T$ share a similar structure with $n=3$ qubits and $\ell=30$ gates; they are constructed as follows. First, a total of $\ell=30$ gates are selected among the set $\{R_X(\theta), R_Y(\theta),R_Z(\theta),\text{CNOT}\}$, where $R(\theta)$ denote 1-qubit rotations with rotation angle $\theta$. The four gates are selected at random with equal probability, and $\theta$ is also a r.v. with uniform distribution in the range of $[0,2\pi)$. Then, the qubits where these gates are applied are selected at random with equal probability from the set $\{1,\ldots,n\}$ for the 1-qubit rotations and from the set $\{(i,j): i,j \in \{1,\ldots,n\}\text{ and }i\neq j\}$ for the CNOT gates. 

    For each testing circuit, the coefficient vector $\hat{\boldsymbol{\alpha}}$ is computed with \eqref{learning_alpha} for three different values of the regularization parameter $\mu$, i.e.,  $\mu\in \{10^{-2},10^{-3},10^{-6}\}$, and using a training set of $S=120$ pairs of $\{W_i,f(W_i)\}_{i=1}^S$. Each pair is constructed as follows. Given the testing circuit $U\in \mathcal{T}_T$, a total of seven rotations are selected at random to remain unchanged. The rest of 1-qubit rotations are substituted with a gate selected from the set $\{I, X, Y, Z, \sqrt{X}, \sqrt{Y}, \sqrt{Z}\}$, where the probability of each gate is proportional to the absolute value of the corresponding coefficient given in the decompositions \eqref{decomposition_X}, \eqref{decomposition_Y}, and \eqref{decomposition_Z}. After this substitution, we efficiently compute the ideal expected value $f(W_i)$ classically, since the $W_i$ circuits are nearly Clifford circuits.

    \begin{figure*}[htb]

  \centering
  \centerline{\includegraphics[width=15cm]{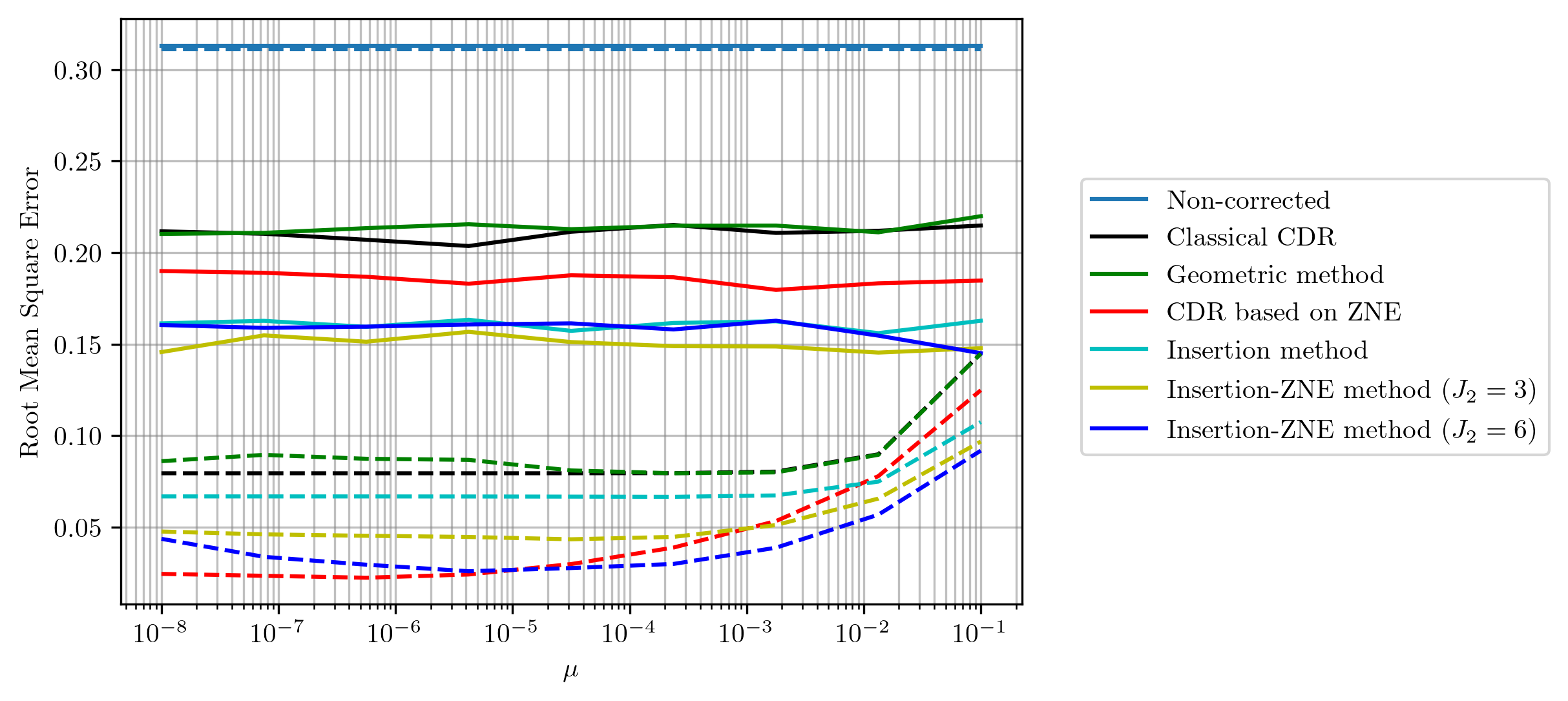}}

\caption{Comparison of the root mean square error of different CDR methods as the regularization parameter changes. The solid lines represent the case where $N = 1000$ samples are used, and the dashed lines represent the ideal case where expected values are used for the computation of the feature maps.}
\label{Figure_reg}
\end{figure*}

    As previously mentioned, we assume that the noise is present in the CNOT gates only. The noisy implementation of these gates is modeled as $\mathcal{N}\circ\,\mathcal{U}_{\text{CNOT}_{i,j}}$, where $\mathcal{N}$ is the noise channel, and $\mathcal{U}_{CNOT_{i,j}}$ denotes a noise-free CNOT gate applied to qubits $i$ and $j$. The noise channel $\mathcal{N}$ has the form $\mathcal{N}=\mathcal{N}_p^{(j)}\circ\mathcal{N}_p^{(i)}$, where $\mathcal{N}_p^{(i)}$ denotes a 1-qubit depolarizing channel with parameter $p$ applied to the $i^{th}$ qubit. This parameter is set to $0.1$ throughout the numerical experiments.

    For the implementation of the insertion method and the insertion-ZNE method, we use the unitary $V=R_X(\pi/8)\otimes I \otimes I$, and set the number of replicas of $V$ to $t_i=i-1$. In both the insertion-ZNE method and the ZNE-based method, we utilize a noise-scaling approach known as gate folding. The implementation of this approach is detailed in Appendix \ref{appendix_implementation}.

     Figures \ref{figureA}, \ref{figureB}, and \ref{figureC} show the root mean square error of $\hat{f}(U)$ as a function of the dimension of the feature vector $J+1$, averaged among $1000$ testing circuits. The solid lines represent the case where $N=1000$ samples are used to compute the empirical means, while the dashed lines represent the ideal case where the matrix $\Phi$ and the feature vector $\boldsymbol{\phi}(U)$ are constructed using expected values, i.e., $\Tr \left(O \tilde{\mathcal{P}}_j (U)\left(\ket{0}\bra{0}^{\otimes n}\right) \right)$, which are computed classically. Generally, these values cannot be computed efficiently; however, their computation becomes feasible due to the small size of the circuits used in the simulations. Results are reported for classical CDR \cite{czarnik2021error}, CDR based on ZNE \cite{lowe2021unified}, the geometric method, the insertion method, and for the insertion-ZNE method for $J_2=\{3,6\}$. Additionally, several values of the regularization parameter $\mu$ are considered. As observed, the insertion-ZNE method performs the best in most scenarios, and, as expected, the error decreases as the dimension of the feature vectors increases.

    \begin{figure}[ht]
        
          \centering
          \centerline{\includegraphics[width=11cm]{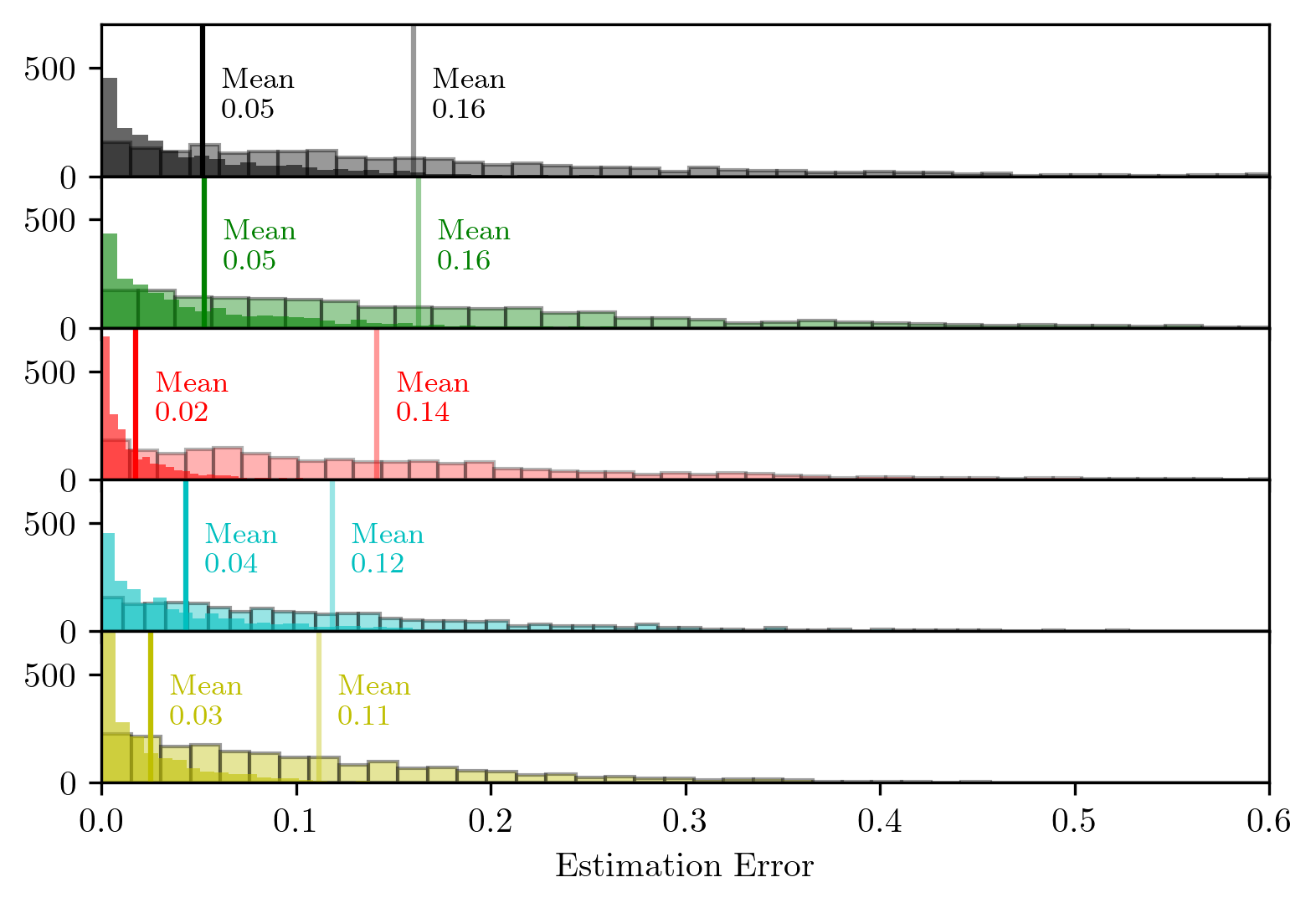}}

          \begin{subfigure}{\textwidth}
          \centering
        
          \includegraphics[width=0.85\linewidth]{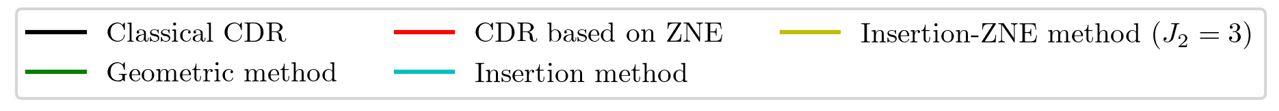}
         \vspace{0.05cm}
        \end{subfigure}

        \caption{Histogram of the estimation error for the different CDR methods.}
        \label{Figure_histograms}
        \end{figure}

    Although similar trends are observed whether ideal expected values or a finite number of samples are used, in the former case, the regularization constant significantly influences the performance of the different methods employed. To examine this dependency, we conducted simulations for different values of the regularization constant, as shown in Figure \ref{Figure_reg}. In the numerical experiments illustrated in Figure \ref{Figure_reg}, we set $J=7$ for the geometric feature map, the insertion feature map, and the ZNE-based feature map, and $J_1=7$ for the insertion-ZNE feature map. The rest of parameters coincide with those used previously. The root mean square error is computed using $2\cdot 10^3$ random testing circuits. For the case with the expected values (dashed-lines), we observe that most methods exhibit a reduced dependency on the regularization parameters for values $\mu<10^{-3}$, with a similar optimal value of the regularization parameter $\mu$ around $[10^{-6},10^{-4}]$. The methods most impacted by the regularization parameter are the ZNE-based method and the insertion-ZNE method with $J_2=6$. Interestingly, these methods achieve the best performance among those considered when using ideal expected values to compute the feature maps. However, while both methods excel in the absence of sampling noise, their performance declines more sharply than that of, for instance, the insertion method once finite sampling (solid lines) is considered. When a finite number of samples is used, as shown also in Figure \ref{Figure_3_num}, the performance is not significantly influenced by the regularization parameter, beyond enhancing the numerical stability of the method. In this scenario, the insertion-ZNE method with $J_2=3$ performs best, even surpassing its performance with $J_2=6$. 
        \begin{figure}[H]
        \centering
        \centerline{\includegraphics[width=15cm]{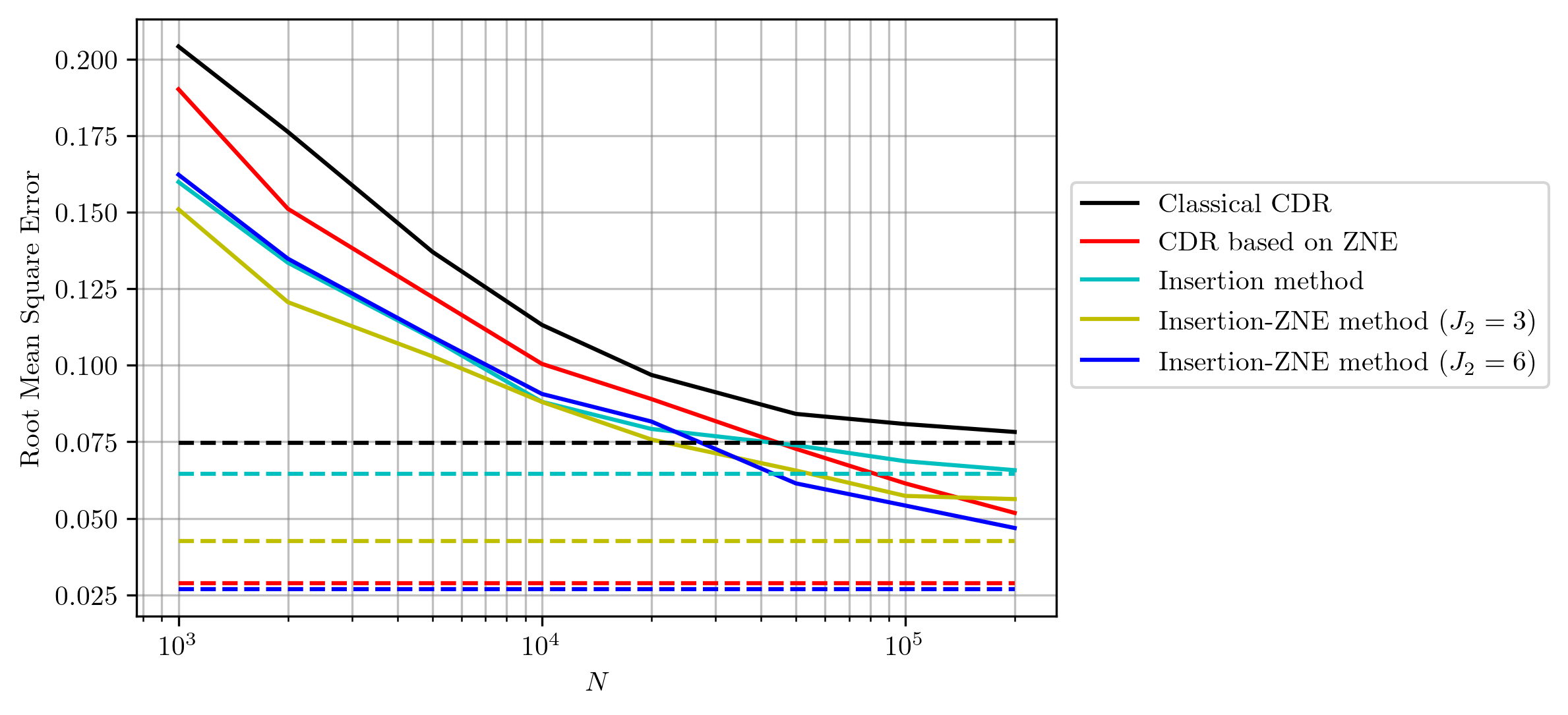}}
        \caption{Comparison of the performance of different CDR methods as the number of samples $N$ changes. The dashed lines represent the ideal case where expected values are used for the computation of the feature maps.}
        \label{Figure_Chan}
    \end{figure}

    All previous simulations use the root mean square error to compare the performance of the different methods. In Figure \ref{Figure_histograms}, for the sake of completeness, we show the histogram of the estimation error for 2000 testing circuits with $\mu=3 \cdot 10^{-5}$. Specifically, in this figure, two different histograms for each method are plotted: the darker one represents the case where ideal expected values are used in the computation of the feature map, and the lighter one represents the case with $N=1000$. We observe that the mean accurately captures the performance of the different methods, and similar trends to those shown in Figure \ref{Figure_3_num} are observed.

    Next, Figure \ref{Figure_Chan} shows the evolution of the RMSE as the number of samples, $N$, increases. Notably, while the ZNE-based approach performs similarly to the insertion-ZNE method with $J_2=6$ when expected values are used (dashed lines), it only starts to surpass the performance of the insertion-ZNE method with $J_2=3$ when $N$ exceeds $10^5$. This observation underscores that both the insertion method and the insertion-ZNE method are more effective at managing the noise when a finite number of samples is used. It is also worth noting that, until convergence, the RMSE decreases at roughly the same rate for all methods.

    \subsection{Numerical experiments on the QFT circuit}
    In these numerical experiments, we test the different methods using non-random circuits, specifically circuits of the form $U = F_n H^{\otimes\, n}$ for different values of $n\in \mathbb{N}$. Here, $F_n$ denotes the unitary matrix associated with the quantum Fourier transform (QFT) on $n$ qubits, and $H$ denotes the Hadamard gate. The implementation of unitary $F_n$ follows the structure shown in \cite{nielsen2010quantum}. However, we substitute the controlled phase shift gates by the circuit shown in Figure \ref{fig:QFT_figure}, which is equivalent up to a global phase. After this substitution, a circuit consisting only of rotations and Clifford gates is obtained.

    \begin{figure}[H]    
    	\centering
    	\includegraphics[width=11cm]{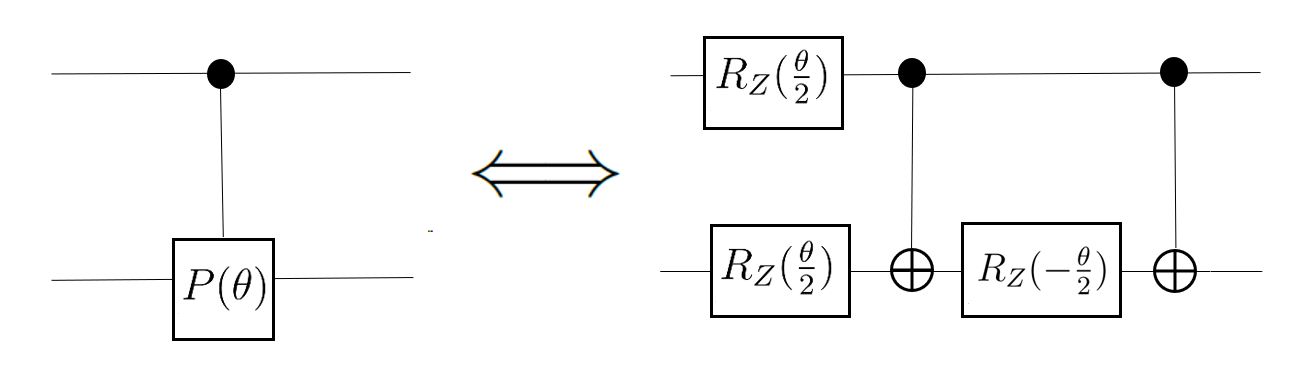}
    	\caption{Substitution of the control phase shift gate.}
    	\label{fig:QFT_figure}
    \end{figure}

    In Figure \ref{fig:QFT_figure_2}, a comparison of the different methods is presented. In particular, this figure shows the average value of $\hat{f}(U)$ over $10$ realizations for different CDR methods and various values of $n$ and $p$. A value of $\mu = 10^{-5}$ is used for all methods to obtain $\hat{\boldsymbol{\alpha}}$, with a training set size of $S=500$. The procedure for generating the training set is the same as the one explained previously. The remaining parameters coincide with those used in the previous experiments, namely, $\theta = \pi/8$, and $J=J_1=7$.
    
     \begin{figure}[H]
        \begin{subfigure}[b]{0.5\textwidth}
          \centering
          \centerline{\includegraphics[width=6.8cm]{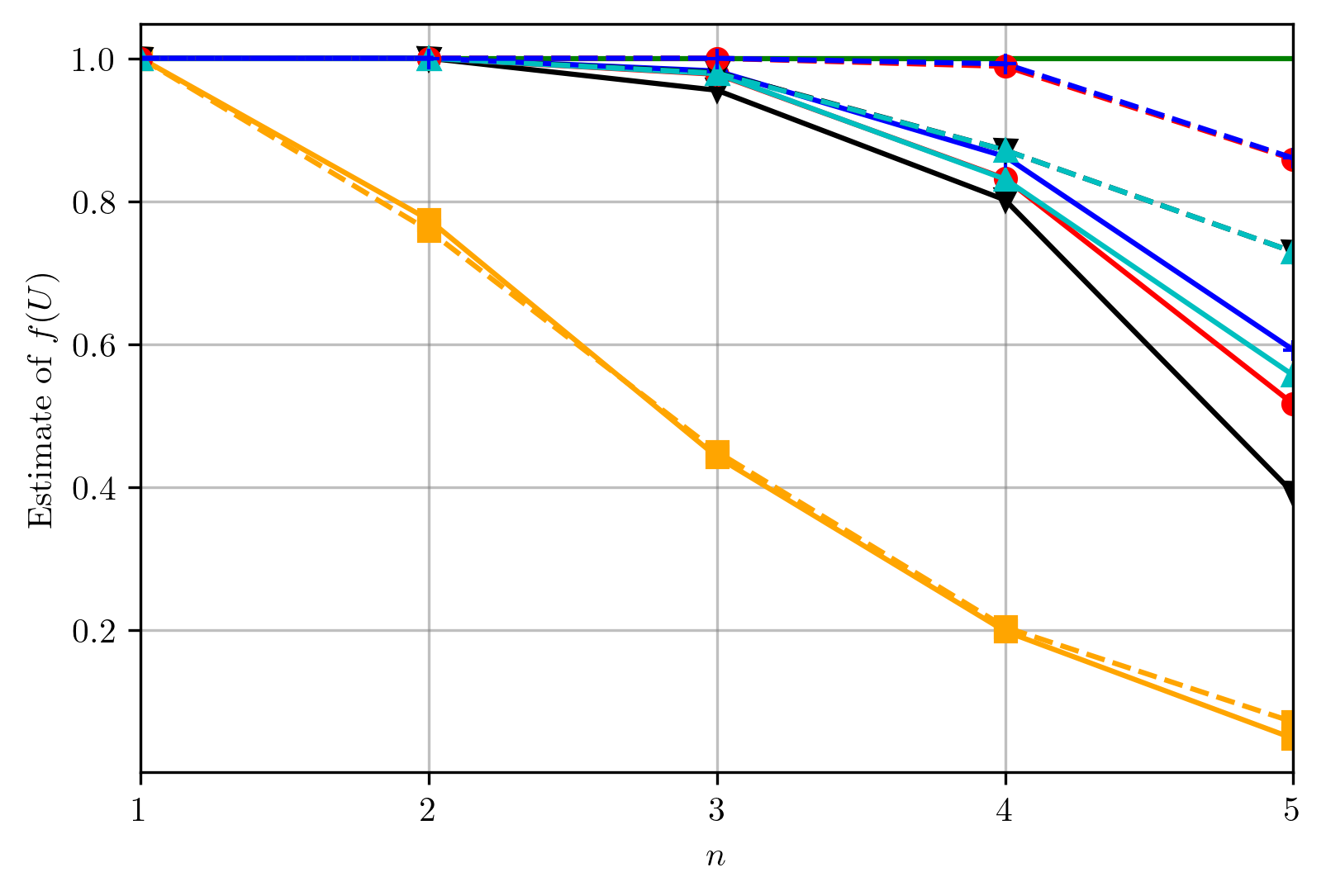}}
          \caption{}
          \medskip
        \end{subfigure}
        \begin{subfigure}[b]{0.5\textwidth}
          \centering
          \centerline{\includegraphics[width=6.8cm]{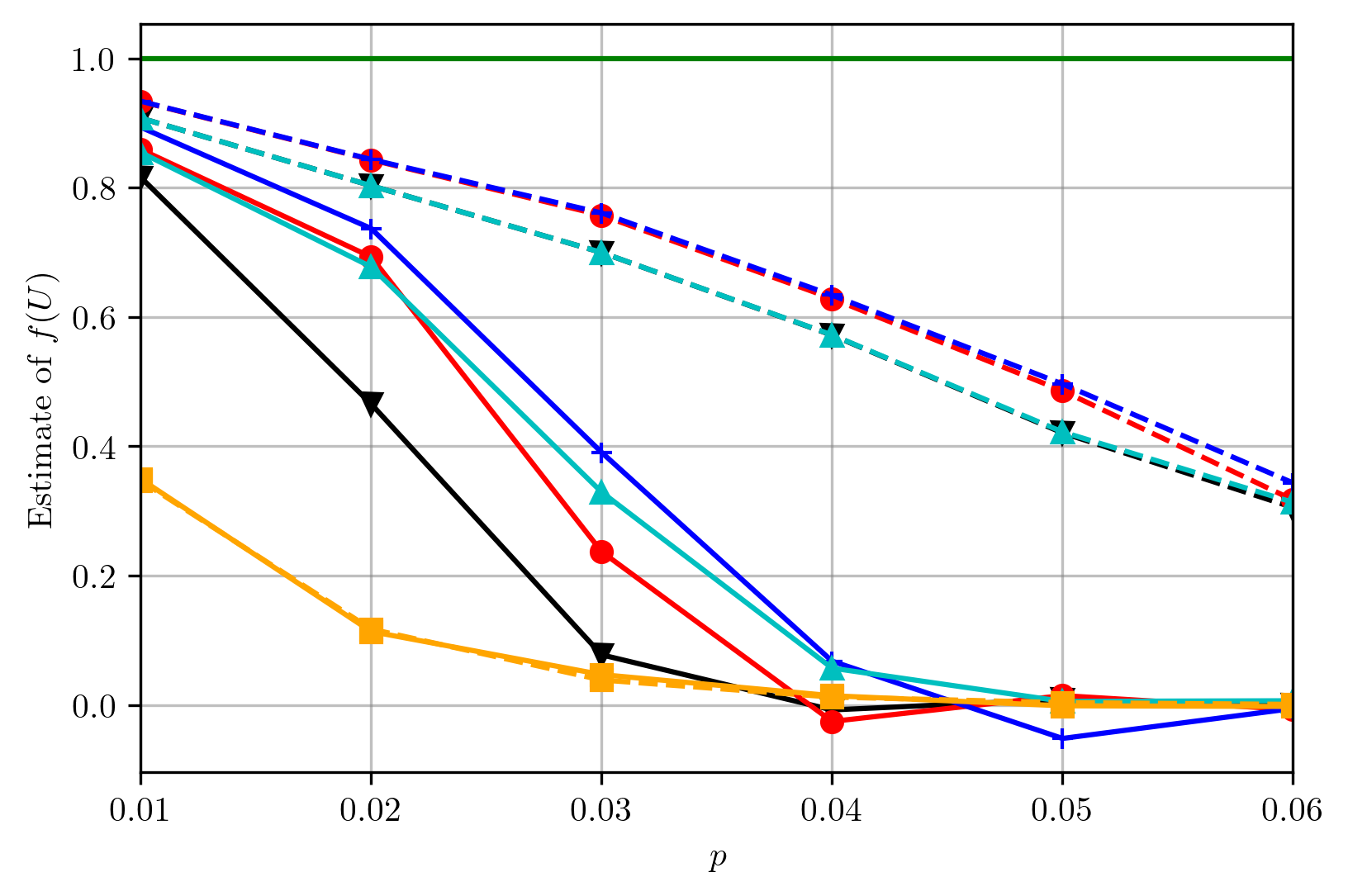}}
          \caption{}
           \label{regularization_change_inf}
           \medskip
        \end{subfigure}
        \begin{subfigure}[b]{\textwidth}
  \centering
  \centerline{\includegraphics[width=9.5cm]{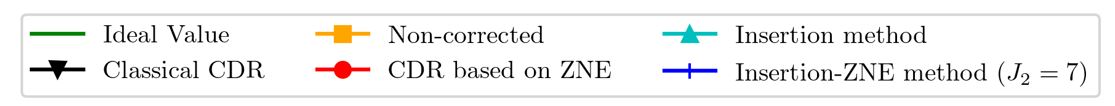}}
   \medskip
\end{subfigure}

        \caption{Comparison of the estimation performed by the different methods studied for the QFT circuit applied after a Hadamard layer of gates. The solid lines represent the case $N=1000$, and the dashed lines the ideal case where expected values are used. In Figure (a), the parameter $p$ that models the noise is set to 0.05; in (b), $n$ is fixed to $7$.}
    	\label{fig:QFT_figure_2}
    \end{figure}
In both figures, the insertion-ZNE method with $J_2 = 7$ performs the best in the scenario of $N = 1000$ samples, followed by the insertion method. Interestingly, when ideal expected values are used, both the ZNE-based approach and the insertion-ZNE method with $J_2 = 7$ achieve similar results, both outperforming the other methods. For Figure (a), which uses $p = 0.05$, we observe that the methods can mitigate the error quite reliably up to $n = 4$ qubits. However, for $n = 5$, the performance degrades significantly, even in the case that uses the ideal expected values. In Figure (b), we observe that for $n = 7$, the performance of all methods degrades slowly with increasing noise levels when ideal expected values are used. In contrast, for $N = 1000$, the performance drops significantly faster.

    \section{Conclusions}
    \label{sec:conclusions}

    In summary, this work introduces two feature maps, extending upon the ones employed in prior research. We provide theoretical justification for these augmented feature vectors and evaluate their performance through numerical experiments. Furthermore, we conduct a theoretical analysis of various characteristics of the resulting learning procedures, determining the scaling of both the complexity and the estimation error. Additionally, we also explore, from an information-theoretical perspective, lower bounds for the different parameters as a function of the noise level in two noise models.

    The first feature vector proposed, i.e., the geometric feature map, does not provide a noticeable performance enhancement. This can be primarily attributed to the inherent noise affecting the values $\hat{\phi}(\tilde{\mathcal{U}}^j)$. In contrast, the insertion method, specifically designed to address and mitigate this challenge, exhibits a reduction in root mean square error for the noise model used. Furthermore, as shown in the simulations, the insertion feature map can be enhanced by using different levels of noise. Notably, both the insertion and insertion-ZNE feature maps exhibit considerable resilience to the effects of using finite samples to estimate the different components of the feature map, as shown in the numerical experiments. This contrasts with the ZNE-based approach, whose performance is significantly impacted by the number of samples used.

    Finally, it is worth noting that although our proposals can reduce the estimation error, exploring other specific instances of the insertion method or other variant with multiple insertions may be interesting to determine the optimal approach for designing these feature maps for CDR. Another potential direction for future research would be to analyze non-linear models for processing these feature maps or to develop feature maps specifically tailored for non-linear models.

    \section{Acknowledgements}

    This work has been funded by grants PID2022-137099NB-C41, PID2019-104958RB-C41 funded by MCIN/AEI/10.13039/501100011033 and FSE+ and by grant 2021 SGR 01033 funded by AGAUR, Dept. de Recerca i Universitats de la Generalitat de Catalunya 10.13039/501100002809.

\bibliographystyle{IEEEbib}
\bibliography{bibliography.bib}

\begin{appendices}

\setcounter{equation}{0}
\renewcommand{\theequation}{\thesection.\arabic{equation}}

    \section{Implementation of the ZNE-based method}\label{appendix_implementation}

    As presented in Section \ref{sec:feature_map}, for the implementation of the ZNE-based method, we use noise factors $\lambda_i = 1 + \frac{2(i-1)}{k}$, where $k$ denotes the number of CNOT gates in circuit $U$. To implement these noise factors, we use the gate folding technique on CNOT gates. Therefore, as in \cite{lowe2021unified}, we assume that most of the noise is generated by the two-qubit gates.

    This method involves inserting an even number of CNOT gates after each CNOT in circuit $U$. Specifically, for $\tilde{\mathcal{P}}(U,\lambda_1)$, no gates are introduced, i.e., it corresponds to $\tilde{\mathcal{U}}$. $\tilde{\mathcal{P}}(U,\lambda_2)$ adds two additional CNOT gates after the first CNOT; $\tilde{\mathcal{P}}(U,\lambda_3)$  builds on $\tilde{\mathcal{P}}(U,\lambda_2)$ by inserting two extra CNOT gates after the second CNOT in the original circuit; and so on, until each CNOT gate in $U$ has two extra CNOT gates added after it. The procedure is repeated, recursively adding two more CNOT gates after each CNOT in the original circuit, until the intended dimension of the feature map is achieved.

    \begin{figure*}[h]
        \centering
        \centerline{\includegraphics[width=12cm]{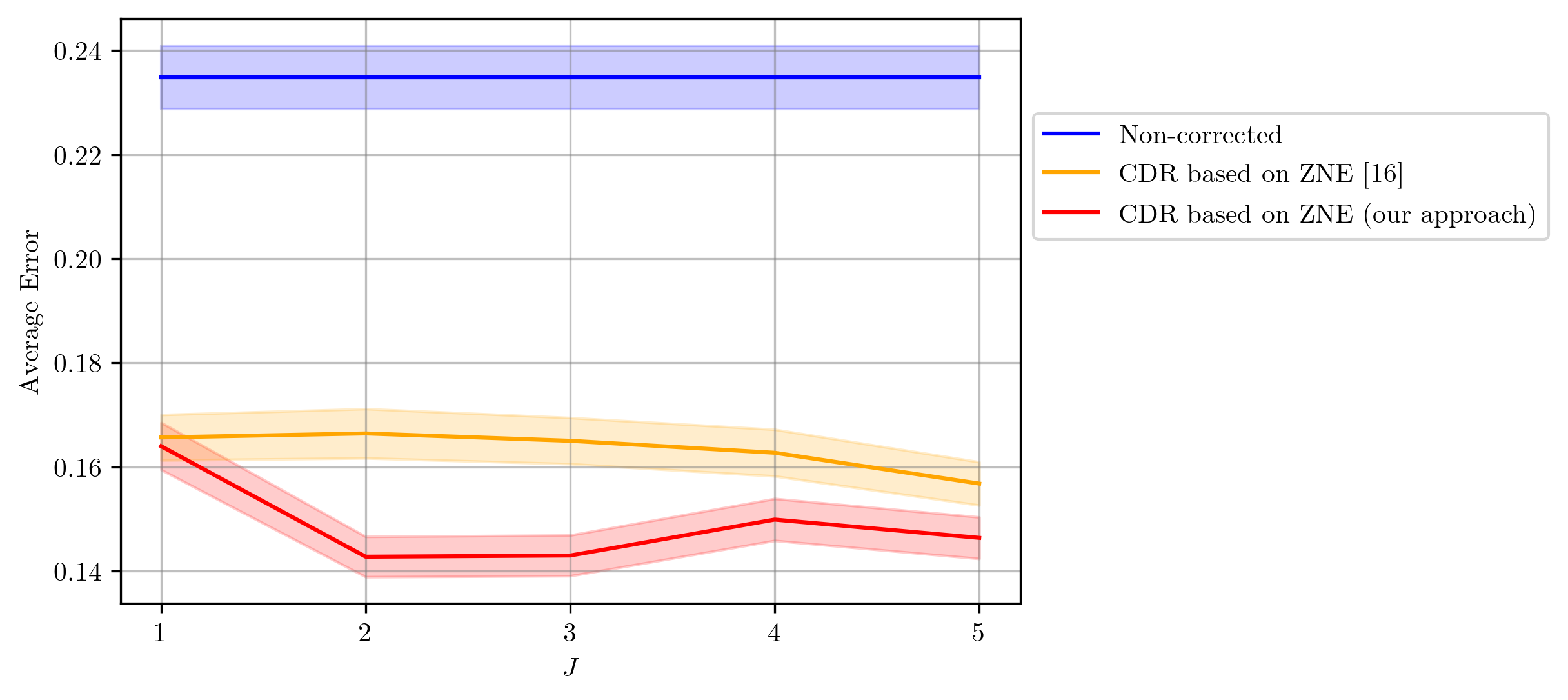}}
        \caption{Comparison of the average error for two ZNE-based methods as the feature map size varies.}
        \label{Figure_comparison_ZNE}
    \end{figure*}

    This implementation produces a small improvement compared to  \cite{lowe2021unified}, as shown in Figure \ref{Figure_comparison_ZNE}. This figure compares the two implementations by plotting the average error across 1000 testing circuits.

    The testing circuits are randomly generated and consist of 30 gates on 3 qubits. The number of samples used to estimate the expected values of the feature map is $N = 1000$. Noise is introduced only in the CNOT gates, where 1-qubit depolarizing channels with parameter $p=0.1$ are applied to the qubits affected by the CNOT gates. For a detailed description of the rest of the implementation, see Section \ref{sec:numerical_experiments}.

\setcounter{equation}{0}
    \section{Proof of Proposition 1}\label{first_appendix}

    In this appendix, we show the proof of Proposition \ref{corollary_1}, which is based on two auxiliary results, Lemma \ref{lemma_} and Lemma \ref{lemma_lip}.  The proofs of these lemmas are provided in sections \ref{appendix_A} and \ref{appendix_lipschitz}, respectively. Both lemmas are stated below.

    \begin{lemma}\label{lemma_}
        For any unitary $U$ and $p\in \mathbb{P}$, where $\mathbb{P}$ denotes the set of prime numbers, there exists another unitary $Y$ such that $g(t,Y)$ is periodic with period $p$ or constant, and 
        \begin{equation*}
             \left\| Y^t \ket{\psi}\bra{\psi}(Y^t)^\dagger-U^t \ket{\psi}\bra{\psi}(U^t)^\dagger \right\|_1 \leq 2\sin \left(\frac{2\pi t}{p}\right)
        \end{equation*}
        for any $t\in[0,\frac{p}{4}]$ and state $\ket{\psi}$.
    \end{lemma}
    \begin{lemma}\label{lemma_lip}
        Function $g(t,U)$ is Lipschitz continuous with Lipschitz constant 
        \begin{equation}
            L\leq 2\left\| O \right\|_2 \sqrt{\left(4 \pi^2+ \frac{16}{3}\pi^4\right)}
        \end{equation}
        that is, 
        \begin{equation}
            |g(t_1,U)-g(t_2,U)|\leq 2\left\| O \right\|_2 \sqrt{\left(4 \pi^2+ \frac{16}{3}\pi^4\right)} \,\,|t_1-t_2|
        \end{equation}
    \end{lemma}

    We start the proof by noting that for any observable $O$ and quantum states $\rho_1$ and $\rho_2$,
        \begin{equation}
            \left| \Tr(O\rho_1)-\Tr(O\rho_2)\right|\leq \left \| O \right \|_2 \left \| \rho_1-\rho_2 \right \|_1
        \end{equation}
    which follows from Eq. (9.42) in \cite{wilde2013quantum}. Combining this result with Lemma \ref{lemma_} yields
    \begin{equation}\label{eq_app_1_1}
        |g(t,U)-g(t,Y)|\leq 2 \left  \| O \right \|_2 \, \sin \left(\frac{2\pi t}{p}\right)
    \end{equation}
    Next, we use the fact that for any periodic function $h(t)$ with period $T$ that satisfies $|h(t_1)-h(t_2)|\leq L |t_1-t_2|$, i.e., $h(t)$ is Lipschitz continuous with Lipschitz constant $L$, there exist complex numbers $c_q = c_{-q}^*$ such that
    \begin{equation}\label{inequality_bound_error_fourier}
        \left|h(t)-\sum_{q=-Q}^Q \,c_q e^{-j\frac{2\pi}{T}qt}\right|\leq \frac{3TL}{2\pi Q} 
    \end{equation}
    This result is a rephrasing of Theorem 1 from \cite{jackson1930theory}.

    Applying this theorem to the function $g(t,Y)$, which is a periodic function with period $p$ and is Lipschitz continuous (Lemma \ref{lemma_lip}), yields that there exist complex numbers $c_q$ such that
    \begin{equation}\label{eq_app_1_2}
        \left|g(t,Y)-\sum_{q=-Q}^Q \,c_q e^{-j\frac{2\pi}{p}qt}\right|\leq \frac{3pL}{2\pi Q} 
    \end{equation}
    Therefore, combining \eqref{eq_app_1_1} and \eqref{eq_app_1_2},
    \begin{equation}\label{inequality_prop}
            \left|g(t,U)-\sum_{q=-Q}^Q c_q \, e^{-j \frac{2\pi}{p}q t}\right|\leq  2\left\| O \right\|_2 \left(C \,   \frac{p}{{Q}} +  \sin \left(\frac{2\pi t}{p}\right)\right),
    \end{equation}
    where $C\leq  3 \sqrt{\left( 1+ \frac{4}{3}\pi^2\right)}$.

    Finally, as the function $g(t,Y)$ has frequencies only within the range $|f|<1$, it follows that for $Q\geq p$, there exist coefficients such that the left-hand side of \eqref{eq_app_1_2} is equal to zero. In other words, the bound provided by \eqref{inequality_bound_error_fourier} can be enhanced in this scenario to $\frac{3pL}{2\pi Q}\cdot \mathbbm{1}\{Q<p\}$.

    \hspace*{16cm}  $\square$ 

\begin{figure*}[h]

\begin{subfigure}[b]{0.5\textwidth}
  \centering
  \centerline{\includegraphics[width=6.8cm]{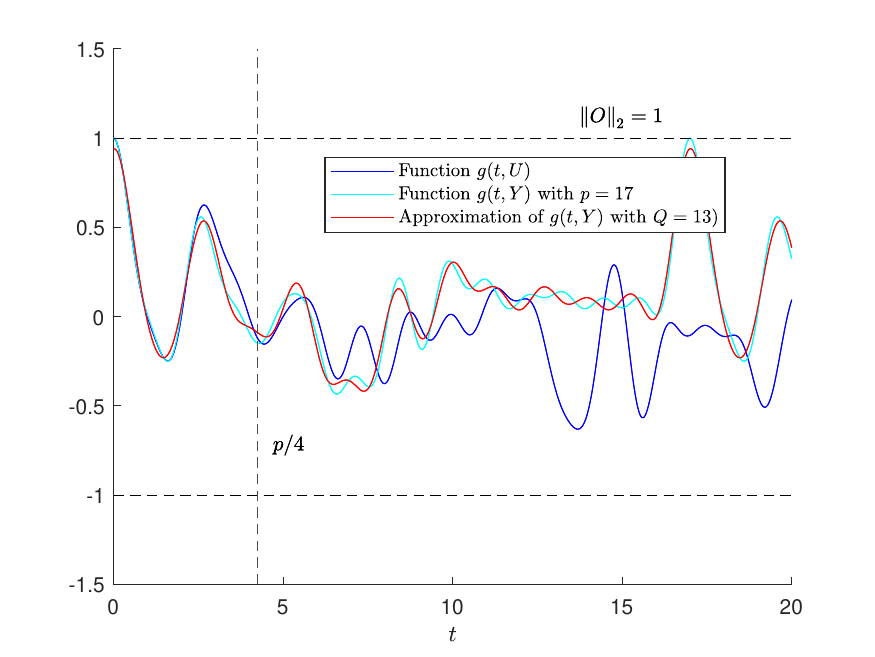}}
  \caption{}\label{figure_1}\medskip
\end{subfigure}
\hfill
\begin{subfigure}[b]{0.5\textwidth}
  \centering
  \centerline{\includegraphics[width=6.8cm]{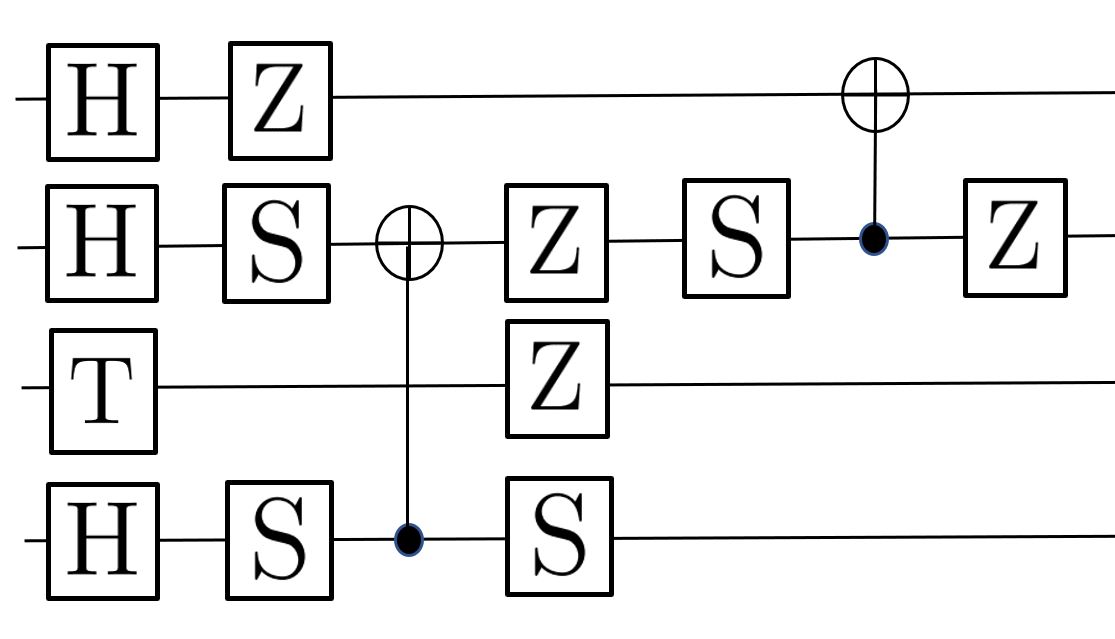}}
  \caption{}\medskip
\end{subfigure}

\label{fig:three graphs}

\caption{Figure (a) shows functions $g(t,U)$, $g(t,Y)$, and the approximation presented in Proposition \ref{corollary_1}. Figure (b) depicts the circuit implementation of the unitary $U$.}
\end{figure*}

    In Figure \ref{figure_1}, we can observe an example of the different functions discussed before. Specifically, we can see that the bound given in Proposition  \ref{corollary_1} is not very tight, as the approximation remains accurate even beyond $t=p/4$.

    \subsection{Proof of Lemma 1}\label{appendix_A}

    In this subsection, we show the proof of Lemma \ref{lemma_}, which we restate here for completeness.

    \begin{lemma**}
        For any unitary $U$ and $p\in \mathbb{P}$, where $\mathbb{P}$ denotes the set of prime numbers, there exists another unitary $Y$ such that $g(t,Y)$ is periodic with period $p$ or constant, and 
        \begin{equation*}
             \left\| Y^t \ket{\psi}\bra{\psi}(Y^t)^\dagger-U^t \ket{\psi}\bra{\psi}(U^t)^\dagger \right\|_1 \leq 2\sin \left(\frac{2\pi t}{p}\right)
        \end{equation*}
        for any $t\in[0,\frac{p}{4}]$ and state $\ket{\psi}$.
    \end{lemma**}

    \begin{proof}
    Let $Y := \sum_{i=1}^{2^n} e^{j 2\pi \frac{\left \lfloor p\omega_i \right \rfloor}{p}}\ket{u_i}\bra{u_i}$, where $\ket{u_i}$ denotes the $i^{th}$ eigenvector of $U$ and $\omega_i$ denotes the phase of the corresponding eigenvalue. From this, it follows that $g(t,Y)$ is periodic with period $p$ if at least one of the phases $\omega_i > \frac{1}{p}$; otherwise, $g(t,Y)$ is constant, as $Y$ equals the identity.
    
    Next, to bound the trace distance, we start by using the identity $\left \| \rho-\sigma \right \|_1=2\sqrt{1-F(\rho,\sigma)}$, which holds as long as both $\rho$ and $\sigma$ are pure states.
    \begin{equation}
        \left \| U^t \ket{\psi}\bra{\psi}(U^t)^\dagger- Y^t\ket{\psi}\bra{\psi} (Y^t)^\dagger \right \|_1  =2 \sqrt{1-\left| \bra{\psi} U^t (Y^t)^\dagger \ket{\psi}\right|^2}
    \end{equation}
    Substituting $Y$ and the spectral decomposition of $U$, 
    \begin{equation*}
        \left| \bra{\psi} U^t (Y^t)^\dagger \ket{\psi}\right|^2=\left|\sum_{i=1}^{2^n}e^{j 2\pi t \left(\omega_i -\frac{\left \lfloor  p\omega_i\right \rfloor }{p}\right)}\left|\bra{u_i}\ket{\psi}\right|^2\right|^2 
    \end{equation*}
    \begin{equation}\label{eq_fidelity}
        =\left(\sum_{i=1}^{2^n} \cos \left( 2\pi t \left(\omega_i -\frac{\left \lfloor  p\omega_i\right \rfloor }{p}\right)\right)\left|\bra{u_i}\ket{\psi}\right|^2\right)^2+\left(\sum_{i=1}^{2^n} \sin \left( 2\pi t \left(\omega_i -\frac{\left \lfloor  p\omega_i\right \rfloor }{p}\right)\right)\left|\bra{u_i}\ket{\psi}\right|^2\right)^2
    \end{equation}
    Note that $0\leq 2\pi t \left(\omega_i -\frac{\left \lfloor  p\omega_i\right \rfloor }{p}\right)\leq \frac{2\pi t}{p}\leq \frac{\pi}{2}$ for $4t\leq p$. Therefore, since $\cos(x)$ is a decreasing function in the interval $[0,\frac{\pi}{2}],$
    \begin{equation}
        \cos \left( 2\pi t \left(\omega_i -\frac{\left \lfloor  p\omega_i\right \rfloor }{p}\right)\right)\geq \cos \left(\frac{2\pi t}{p}\right)
    \end{equation}
    Similarly,
    \begin{equation}
        \sin \left( 2\pi t \left(\omega_i -\frac{\left \lfloor  p\omega_i\right \rfloor }{p}\right)\right)\geq 0
    \end{equation}
    Substituting both results in \eqref{eq_fidelity}, it follows that 
    \begin{equation}
        \left| \bra{\psi} U^t (Y^t)^\dagger \ket{\psi}\right|^2 \geq \cos^2 \left(\frac{2\pi t}{p}\right)= 1- \sin^2\left(\frac{2\pi t}{p}\right)
    \end{equation}
    which implies
    \begin{equation}
        \left \| U^t \ket{\psi}\bra{\psi}(U^t)^\dagger- Y^t\ket{\psi}\bra{\psi} (Y^t)^\dagger \right \|_1  \leq  2\sin \left( \frac{2\pi t}{p}\right)
    \end{equation}
   for $t\leq p/4$.
    \end{proof}

    \subsection{Lipschitz continuity of $g(t,U)$}\label{appendix_lipschitz}

        In this subsection, we prove that the function $g(t,U)$ is Lipschitz continuous and provide an upper bound for its Lipschitz constant. The formal statement is shown below.

        \begin{lemma*}
        Function $g(t,U)$ is Lipschitz continuous with Lipschitz constant 
        \begin{equation}
            L\leq 2\left\| O \right\|_2 \sqrt{\left(4 \pi^2+ \frac{16}{3}\pi^4\right)}
        \end{equation}
        that is, 
        \begin{equation}
            |g(t_1,U)-g(t_2,U)|\leq 2\left\| O \right\|_2 \sqrt{\left(4 \pi^2+ \frac{16}{3}\pi^4\right)} \,\,|t_1-t_2|
        \end{equation}
    \end{lemma*}

    \begin{proof}

    From the definition of function $g(t,U)$, it follows that
    \begin{equation}
        |g(t_1,U)-g(t_2,U)|= \left| \Tr(O(\rho_{t_1}-\rho_{t_2})) \right|\leq  \left\| O \right\|_ 2  \left\| \rho_{t_1}-\rho_{t_2} \right\|_1
    \end{equation}
    where $\rho_t= (U^t)^\dagger (\ket{0}\bra{0})^{\otimes n} U^t$, and $\left\| O \right\|_ 2$ denotes the spectral norm of observable $O$. The inequality follows from Eq. (9.42) in \cite{wilde2013quantum}. Since both states are pure, the trace distance can be expressed as
    \begin{equation}
        \left\| \rho_{t_1}-\rho_{t_2} \right\|_1= 2\sqrt{1-\left|\bra{0}^{\otimes n} (U^{t_1} )^{\dagger} U^{t_2}\ket{0}^{\otimes n}\right|^2}
    \end{equation}
    Next, we bound the fidelity. However, before doing so, we need to expand its expression as shown below:
    \begin{align}
        F(\rho_{t_1}, \rho_{t_2})&=\left| \bra{0}^{\otimes n} (U^{t_1})^\dagger U^{t_2} \ket{0}^{\otimes n}\right|^2 = \left|\bra{0}^{\otimes n} U^{t_2-t_1} \ket{0}^{\otimes n}\right|^2\nonumber \\&= \left | \sum_{i=1}^{2^n} e^{j 2\pi \omega_i (t_2-t_1)} |\bra{0} \ket{u_i}|^2\right|^2 \nonumber \\ &= \left( \sum_{i=1}^{2^n} \cos( 2\pi \omega_i (t_2-t_1))|\bra{0} \ket{u_i}|^2\right)^2+\left( \sum_{i=1}^{2^n} \sin( 2\pi \omega_i (t_2-t_1))|\bra{0} \ket{u_i}|^2\right)^2 \nonumber \\ &= \left( \sum_{i=1}^{2^n} \cos( 2\pi \omega_i |t_2-t_1|)|\bra{0} \ket{u_i}|^2\right)^2+\left( \sum_{i=1}^{2^n} \sin( 2\pi \omega_i |t_2-t_1|)|\bra{0} \ket{u_i}|^2\right)^2
    \end{align}
    where $\sum_{i=1}^{2^n} e^{j 2\pi \omega_i} \ket{u_i}\bra{u_i}$ denotes the spectral decomposition of unitary $U$. To bound the fidelity, we use the identities $\cos(x)\geq 1-\frac{x^2}{2}$ for all $x \in \mathbb{R}$, and  $\sin(x)\geq x-\frac{x^3}{6}$ for all $x \in \mathbb{R}^+$. However, we cannot apply them directly since $a<b$ does not imply $a^2 <b^2$. For this reason, we need to guarantee that both lower bounds, i.e., $x-\frac{x^3}{6}$ and $1-\frac{x^2}{2}$, are positive for any argument $x=2\pi \omega_i |t_2-t_1|$. This holds if we assume $|t_2-t_1|\leq \frac{1}{\sqrt{2} \pi}$. 

    Using the previous identities, and the assumption regarding $|t_2-t_1|$, we have:
    \begin{align}
        F(\rho_{t_1}, \rho_{t_2}) &\geq \left( 1- 2\pi^2 (t_1-t_2)^2 \chi_2 \right)^2+ \left( 2\pi |t_2-t_1|\chi_1 -\frac{(2\pi|t_2-t_1|)^3}{6} \chi_3\right)^2\nonumber \\ &\geq 1-4\pi^2 (t_1-t_2)^2 \chi_2-\frac{(2\pi (t_2-t_1))^4}{3} \chi_1 \chi_3
    \end{align}
    where $\chi_j:= \sum_{i=1}^{2^n} \omega_i^j |\bra{0} \ket{u_i}|^2$. Using $|t_2-t_1|\leq \frac{1}{\sqrt{2} \pi} <1$, and the inequalities $\chi_j\leq 1$ for $j\in\{1,2,3\}$, it follows that
    \begin{equation}
         F(\rho_{t_1}, \rho_{t_2})\geq 1- \left(4 \pi^2+ \frac{16}{3}\pi^4\right) (t_1-t_2)^2
    \end{equation}
    Therefore, for $|t_2-t_1|\leq \frac{1}{\sqrt{2} \pi}$
    \begin{equation}
        |g(t_1,U)-g(t_2,U)|\leq 2\left\| O \right\|_2 \sqrt{\left(4 \pi^2+ \frac{16}{3}\pi^4\right)} \,\,\,|t_1-t_2|
    \end{equation}

    Finally, note that $|g(t_1,U)-g(t_2,U)|\leq 2 \left\| O \right\|_2$ for any pair $t_1,t_2 \in \mathbb{R}$. Therefore, since 
    \begin{equation}
        \sqrt{\left(4 \pi^2+ \frac{16}{3}\pi^4\right)}|t_1-t_2|>1
    \end{equation}
    holds for $|t_1-t_2|>\frac{1}{\sqrt{2} \pi}$, it follows that 
    \begin{equation}
        |g(t_1,U)-g(t_2,U)|\leq 2\left\| O \right\|_2 \sqrt{\left(4 \pi^2+ \frac{16}{3}\pi^4\right)} \,\,\,|t_1-t_2|
    \end{equation}
    is satisfied for any  pair $t_1,t_2 \in \mathbb{R}$.

    \end{proof}

\setcounter{equation}{0}
    \section{Error Bound }\label{appendix_bound_error}

    In this appendix, we prove inequality \eqref{inequality_error}, i.e.,
    \begin{equation}
        \mathbb{E} \left[|\hat{f}(U)-f(U)|\right] \leq \frac{\|\boldsymbol{\hat{\alpha}}\|_2\left \| O \right \|_2}{\sqrt{N}}+ |f(U)-\boldsymbol{\hat{\alpha}}^T \boldsymbol{\phi}_{\infty}(U)|
    \end{equation}

    \begin{proof}
        
        \begin{align}\label{expected_value_correction}
              \mathbb{E}&\left[\left|\hat{f}(U)-f(U) \right|\right]= \mathbb{E}\left[\left|\hat{f}(U)-B\left(\boldsymbol{\hat{\alpha}}^T \boldsymbol{\phi}_{\infty}(U)\right)+B\left(\boldsymbol{\hat{\alpha}}^T \boldsymbol{\phi}_{\infty}(U)\right)-f(U) \right|\right] \nonumber \\& \leq     \mathbb{E}\left[\left|\hat{f}(U)-B\left(\boldsymbol{\hat{\alpha}}^T \boldsymbol{\phi}_{\infty}(U)\right)\right| \right]+  \left|B\left(\boldsymbol{\hat{\alpha}}^T \boldsymbol{\phi}_{\infty}(U)\right)-f(U) \right|  \nonumber \\ & \leq  \sqrt{\mathbb{E}\left[\left(\hat{f}(U)-B\left(\boldsymbol{\hat{\alpha}}^T \boldsymbol{\phi}_{\infty}(U)\right)\right)^2 \right]}+ \left|B\left(\boldsymbol{\hat{\alpha}}^T \boldsymbol{\phi}_{\infty}(U)\right)-f(U) \right| \nonumber \\& \leq  \sqrt{\mathbb{E}\left[\left(\boldsymbol{\hat{\alpha}}^T \boldsymbol{\phi}(U)-\mathbb{E}[\boldsymbol{\hat{\alpha}}^T \boldsymbol{\phi}(U)]\right)^2 \right]}+  \left|\boldsymbol{\hat{\alpha}}^T \boldsymbol{\phi}_{\infty}(U)-f(U) \right|  \nonumber \\ &  =   \sqrt{\mathrm{Var}[\boldsymbol{\hat{\alpha}}^T \boldsymbol{\phi}(U)]}+  \left|\boldsymbol{\hat{\alpha}}^T \boldsymbol{\phi}_{\infty}(U)-f(U) \right| \nonumber \\ & \leq \frac{\| \boldsymbol{\hat{\alpha}} \|_2\left \| O \right \|_2}{\sqrt{N}}+  \left|\boldsymbol{\hat{\alpha}}^T \boldsymbol{\phi}_{\infty}(U)-f(U) \right| 
        \end{align}

    The second inequality follows from Jensen's inequality, the third one uses the fact that $B(x):=\max \left\{ \min\{x, \left \| O \right \|_2\},-\left \| O \right \|_2\right\} $ is Lipschitz continuous with $L=1$, i.e., $|B(x)-B(x')| \leq |x-x'|$ for any $x,x'\in \mathbb{R}$, and the last inequality uses that $\mathrm{Var}[\boldsymbol{\hat{\alpha}}^T \boldsymbol{\phi}(U)]\leq \frac{\| \boldsymbol{\hat{\alpha}} \|_2^2 \left\| O \right \|_2^2}{N}$ for any $U$. This follows from
    \begin{equation}
        \mathrm{Var}[\boldsymbol{\hat{\alpha}}^T \boldsymbol{\phi}(U)]= \sum_{i=1}^J \hat{\alpha}_i^2 \,\mathrm{Var} [\hat{\phi}(\tilde{\mathcal{P}}_i (U))]= \frac{1}{N} \sum_{i=1}^J \hat{\alpha}_i^2 \,\, \mathrm{Var} [A_i(U)]
    \end{equation}

    where $A_i(U)$ denotes the r.v. obtained by measuring the state $\tilde{\mathcal{P}}_i (U)(\ket{0}\bra{0}^{\otimes n})$ with observable $O$. Next, using that a bounded r.v. $X\in[-C,C]$ satisfies that $\mathrm{Var}[X] \leq C^2$,

    \begin{equation}
         \mathrm{Var}[\boldsymbol{\hat{\alpha}}^T \boldsymbol{\phi}(U)]\leq \frac{\| \boldsymbol{\hat{\alpha}} \|_2^2\left \| O \right \|_2^2}{N}
    \end{equation}

    \end{proof}
    
\setcounter{equation}{0}
    \section{Proof of Theorem \ref{theorem_generalization}}\label{appendix_C}\label{proof_thm_1}

    In this appendix, we show the proof of Theorem \ref{theorem_generalization}, which follows trivially from two auxiliary results, Lemma \ref{lemma_aux_thm_1} and Lemma \ref{lemma_2}.  The statements and proofs of these lemmas are presented below.

    For the first lemma, two additional facts are used. First, note that by substituting the decompositions \eqref{decomposition_X}, \eqref{decomposition_Y}, and \eqref{decomposition_Z}, we can express any unitary $U$ composed of CNOT gates and rotations as a sum of almost Clifford channels, i.e.,
    \begin{equation}\label{decomposition}
        U \rho U^\dagger = \sum_{i=1}^{3^{\ell_r}} w_i(\boldsymbol{\theta}) C_i \rho C_i^{\dagger} 
    \end{equation}
    where each circuit $C_i$ has the same structure as $U$ and is mainly composed by Clifford gates, the coefficients depend on the angles of the different substituted rotations in the circuit $U$, denoted by the vector $\boldsymbol{\theta}$, and $\ell_r$ is the number of substituted rotations in the circuit. For the proof of Lemma \ref{lemma_aux_thm_1}, instead of indexing the circuits with an integer, we use the elements of the set $\mathcal{A}=\{ 1,2,3\}^{\ell_r}$, as this is more convenient. 

    Second, the probability distribution induced by the procedure described in section \ref{circuits_training_set}, $\mathbb{P}_{W\sim D(U)}(C_i)$, can be expressed as $|w_i(\boldsymbol{\theta})|/ \sum_{i=1}^{3^{\ell_r}} |w_i(\boldsymbol{\theta})|$.

    \begin{lemma}\label{lemma_aux_thm_1}
        Assuming that any 1-qubit gate is modified by the same quantum channel $\mathcal{N}$ when implemented on the noisy device. Then, the insertion feature map, the ZNE-based feature map, and the insertion-ZNE feature map satisfy, 
        \begin{equation}
            \boldsymbol{\phi}_\infty (U) =\sum_i w_i(\boldsymbol{\theta}) \boldsymbol{\phi}_\infty (C_i) 
        \end{equation}
        implying, 
        \begin{align}\label{equality_expression}
            |f(U)-\boldsymbol{\hat{\alpha}}^T \boldsymbol{\phi}_{\infty}(U)| &= \left|\sum_{i=1}^{3^{\ell_r}} w_i(\boldsymbol{\theta}) \left( f(C_i)-\boldsymbol{\hat{\alpha}}^T \boldsymbol{\phi}_{\infty}(C_i)\right)\right| \nonumber \\ &=N(\boldsymbol{\theta})\left|\sum_{i=1}^{3^{\ell_r}} \mathbb{P}_{W\sim D(U)}(C_i) \, \mathrm{sign}(w_i(\boldsymbol{\theta})) \left( f(C_i)-\boldsymbol{\hat{\alpha}}^T \boldsymbol{\phi}_{\infty}(C_i)\right)\right|\nonumber \\ &=  N(\boldsymbol{\theta})\left| \mathbb{E}_{W\sim D(U)} \left[s_U(W)\left( f(W)-\boldsymbol{\hat{\alpha}}^T \boldsymbol{\phi}_{\infty}(W)\right)\right]\right| 
        \end{align}
        where $N(\boldsymbol{\theta}):=\sum_{i} |w_i(\boldsymbol{\theta})|$, $\,\mathbb{P}_{W\sim D(U)}(C_i)=|w_i(\boldsymbol{\theta})|/N(\boldsymbol{\theta})$ and $s_U(W):=\mathrm{sign}(w(\boldsymbol{\theta}))$, being $w(\boldsymbol{\theta})$ the coefficient associated with $W$ when decomposing $U$.

    \end{lemma}

    \begin{proof}

    Let $\mathcal{U}=\mathcal{U}_\ell \circ \mathcal{U}_{\ell-1} \circ \cdots \circ \mathcal{U}_1$, where each $\mathcal{U}_\ell$ is either a CNOT gate or a single-qubit rotation. For simplicity in our argument, we assume that each $\mathcal{U}_i$ is a 1-qubit rotation. Therefore, using the decompositions given in equations \eqref{decomposition_X}, \eqref{decomposition_Y}, and \eqref{decomposition_Z}, we have that for all $i\in \{1,\cdots, \ell\}$
    \begin{equation}
        \mathcal{U}_i= \beta_{i,1}\, \mathcal{G}_{i,1}+ \beta_{i,2} \,\mathcal{G}_{i,2}+\beta_{i,3} \, \mathcal{G}_{i,3}
    \end{equation}
    where $\mathcal{G}_{i,j}$ denotes a Clifford gate, and $\beta_{i,j}\in \mathbb{R}$. Next, substituting on the expression for $\mathcal{U}$,
    \begin{align}
        \mathcal{U}&=\sum_{a\in \mathcal{A}} \left(\prod_{j=1}^\ell \beta_{i,a_i} \right) \mathcal{G}_{\ell, a_\ell}\circ \mathcal{G}_{\ell-1, a_{\ell-1}} \circ \cdots \circ \mathcal{G}_{1, a_1} \nonumber \\ &= \sum_{a\in \mathcal{A}} w_a(\boldsymbol{\theta}) \, \mathcal{C}_{a}
    \end{align}    
     Importantly, the same linear decomposition holds if we insert an arbitrary unitary $V$ at the same position $k$ in all the circuits. That is,
    \begin{align}
        \mathcal{U}_\ell \circ \cdots \circ \mathcal{U}_{k} \circ \mathcal{V} \circ \mathcal{U}_{k-1}  \circ \cdots \circ \mathcal{U}_1 &=\sum_{a\in \mathcal{A}} \left(\prod_{j=1}^\ell \beta_{i,a_i} \right) \mathcal{G}_{\ell, a_\ell}\circ \cdots \circ \mathcal{G}_{k, a_{k}} \circ \mathcal{V} \circ \mathcal{G}_{k-1, a_{k-1}} \circ \cdots \circ \mathcal{G}_{1, a_1} \nonumber \\ &= \sum_{a\in \mathcal{A}} w_a(\boldsymbol{\theta}) \, \mathcal{G}_{\ell, a_\ell}\circ \cdots \circ \mathcal{G}_{k, a_{k}} \circ \mathcal{V} \circ \mathcal{G}_{k-1, a_{k-1}} \circ \cdots \circ \mathcal{G}_{1, a_1}
    \end{align}
    Furthermore, since we assume that any 1-qubit gate is modified by the same quantum channel $\mathcal{N}$ when implemented on the noisy device, it follows that
    \begin{align}
        \tilde{\mathcal{U}}_\ell \circ \cdots \circ  \tilde{\mathcal{U}}_{k} \circ  \tilde{\mathcal{V}} \circ  \tilde{\mathcal{U}}_{k-1}  \circ \cdots \circ  \tilde{\mathcal{U}}_1 &= \sum_{a\in \mathcal{A}} w_a(\boldsymbol{\theta}) \,  \tilde{\mathcal{G}}_{\ell, a_\ell}\circ \cdots \circ  \tilde{\mathcal{G}}_{k, a_{k}} \circ  \tilde{\mathcal{V}} \circ  \tilde{\mathcal{G}}_{k-1, a_{k-1}} \circ \cdots \circ  \tilde{\mathcal{G}}_{1, a_1}
    \end{align}   
    where $\tilde{\mathcal{X}}= \mathcal{N}\circ \mathcal{X}$ denotes the noisy implementation of unitary $\mathcal{X}$. This result generalizes to the insertion of multiple gates into the original circuit and, therefore, applies to the feature maps listed above. This implies
    \begin{equation}
        \boldsymbol{\phi}_\infty (U) =\sum_i w_i(\boldsymbol{\theta}) \boldsymbol{\phi}_\infty (C_i) 
    \end{equation}

    \end{proof}

     \begin{lemma}\label{lemma_2}
     Let $W_1, \cdots, W_S$ denote $S$ i.i.d. samples from distribution $\mathbb{P}_{W\sim D(U)}(W)=|w(\boldsymbol{\theta})|/N(\boldsymbol{\theta})$, and $O$ be an observable such that $\left \| O \right \|_2\geq 1$, then
        \begin{align}\label{generalization_modified_app}
            &\Big | \mathbb{E}_{W\sim D(U)}  \Big[s_U(W)\Big( f(W)-\boldsymbol{\alpha}^T \boldsymbol{\phi}_{\infty}(W)\Big)\Big]\Big| \leq \left| \frac{1}{S}\sum_{i=1}^S s_U(W_i) \left( f(W_i)-\boldsymbol{\alpha}^T \boldsymbol{\phi}_{\infty}(W_i)\right)\right| \nonumber \\ & + \frac{\left \lceil \| \boldsymbol{\alpha} \|_2 \right \rceil \left \| O \right \|_2 \sqrt{(J+1)}}{\sqrt{S}} \left( 2+\sqrt{72 \log \frac{8 \left \lceil \| \boldsymbol{\alpha} \|_2 \right \rceil^2}{\delta}}  \right)  
        \end{align}
        holds with probability $1-\delta$ for all values of $\boldsymbol{\alpha}\in \mathbb{R}^{J+1} \backslash  \{\boldsymbol{0}\}$ over the random samples $W_i$.
    \end{lemma}

    \begin{proof}
        First, using Theorem 3.3 \cite{mohri2018foundations},
        \begin{align}
            \mathbb{E}_{W\sim D(U) }\left[\frac{s_U(W)\left(f(W)-\boldsymbol{\alpha}^T \boldsymbol{\phi}_{\infty}(W) \right)}{2\left \| O \right \|_2 (1+ \Lambda \sqrt{J+1})}+\frac{1}{2}\right] &\leq \frac{1}{S}\sum_{i=1}^S \left( \frac{s_U(W_i)\left(f(W_i)-\boldsymbol{\alpha}^T \boldsymbol{\phi}_{\infty}(W_i) \right)}{2\left \| O \right \|_2 (1+ \Lambda \sqrt{J+1})}+\frac{1}{2} \right)\nonumber \\ &+2\hat{\mathcal{R}}_\mathcal{T} (\mathcal{G}_\Lambda)+3 \sqrt{\frac{1}{2S}\ln  \left(\frac{2}{\delta} \right)}
        \end{align}
        holds with at least probability $1-\delta$ for all $\|\boldsymbol{\alpha}\|_2\leq \Lambda$, where $\hat{\mathcal{R}}_\mathcal{T}(\cdot)$ denotes the empirical Rademacher complexity, $\mathcal{T}=\{W_1,\cdots,W_S\}$, and 
        \begin{equation}
        \mathcal{G}_\Lambda:=\left\{\frac{s_U(W)\left(f(W)-\boldsymbol{\alpha}^T \boldsymbol{\phi}_{\infty}(W) \right)}{2\left \| O \right \|_2 (1+ \Lambda \sqrt{J+1})}+\frac{1}{2}:\| \boldsymbol{\alpha} \|_2 \leq \Lambda \right\}
        \end{equation}
        Note that to apply Theorem 3.3 \cite{mohri2018foundations}, we need to use functions $g:\mathcal{W} \rightarrow [0,1]$. For this reason, the factor $2\left \| O \right \|_2 (1+ \Lambda \sqrt{J+1})$ is added. In particular, this factor can be derived from
        \begin{align}
            \left| s_U(W) (f(W)-\boldsymbol{\alpha}^T \boldsymbol{\phi}_{\infty}(W)\right| &\leq |f(W)|+|\boldsymbol{\alpha}^T \boldsymbol{\phi}_{\infty}(W))|  \nonumber \\ &\leq  \left \| O \right \|_2 + \left \|\boldsymbol{\alpha} \right \|_2  \left \| \boldsymbol{\phi}_{\infty}(W) \right \|_2  \nonumber \\ & \leq \left \| O \right \|_2+ \Lambda \sqrt{\left \| O \right \|_2^2 J +1} \nonumber \\ & \leq \left \| O \right \|_2 \left ( 1+\Lambda \sqrt{ J +1} \right)
        \end{align}
        where the second inequality follows from $\Tr(O \rho) \leq \left \| O \right \|_2$ and Cauchy-Schwarz inequality. The third inequality uses the definition of  $\boldsymbol{\phi}_{\infty}(W)$, and the last inequality relies on the assumption $\left \| O \right \|_2\geq 1$.

        Next, using the fact that for any set $\mathcal{A}=\{h\circ f (x): f\in \mathcal{B}\}$, where $h$ is a Lipschitz continuous function with Lipschitz constant $L$,  
        \begin{equation}
             \hat{\mathcal{R}}_\mathcal{T} (\mathcal{A}) \leq L \hat{\mathcal{R}}_\mathcal{T} (\mathcal{B}) 
        \end{equation}
        we have that 
        \begin{equation}
             \hat{\mathcal{R}}_\mathcal{T} (\mathcal{G}_\Lambda) \leq \frac{\hat{\mathcal{R}}_\mathcal{T}(\mathcal{H}_\Lambda)}{2\left \| O \right \|_2 (1+ \Lambda \sqrt{J+1}) }  
        \end{equation}
        where $\mathcal{H}_\Lambda:=\{s_U(W) (f(W)-\boldsymbol{\alpha}^T \boldsymbol{\phi}_{\infty}(W): \| \boldsymbol{\alpha} \|_2 \leq \Lambda \} $.
        Furthermore, using Talagrand’s contraction lemma (Lemma 5.7 \cite{mohri2018foundations}) and Theorem 6.12 \cite{mohri2018foundations},
        \begin{align}
            \hat{\mathcal{R}}_\mathcal{T} (\mathcal{H}_\Lambda) &\leq \frac{\Lambda \sqrt{\Tr(K)}}{S} \nonumber \\& = \frac{\Lambda \sqrt{\sum_{i=1}^S \boldsymbol{\phi}_{\infty}(W_i)^T \boldsymbol{\phi}_{\infty}(W_i)}}{S} \nonumber \\& \leq \frac{\Lambda \sqrt{ S \left( \left \| O \right \|_2^2 J +1 \right)}}{S} \nonumber \\& \leq \frac{\Lambda \left \| O \right \|_2 \sqrt{  J +1 }}{\sqrt{S}}
        \end{align}
        where the first inequality utilizes $\|\boldsymbol{\phi}_{\infty}(W_i)\|_2 \leq \sqrt{\left \| O \right \|_2^2 J +1}$, and the last inequality relies on the assumption $\left \| O \right \|_2\geq 1$. Therefore, by substituting and rearranging terms,
        \begin{align}
            \mathbb{E}_{W\sim D(U) }&\left[s_U(W)\left(f(W)-\boldsymbol{\alpha}^T \boldsymbol{\phi}_{\infty}(W) \right)\right] \leq \frac{1}{S}\sum_{i=1}^S  s_U(W_i)\left(f(W_i)-\boldsymbol{\alpha}^T \boldsymbol{\phi}_{\infty}(W_i)  \right)\nonumber \\ &+\frac{2\Lambda \left \| O \right \|_2 \sqrt{  J +1 }}{\sqrt{S}}+6\left \| O \right \|_2 (1+ \Lambda \sqrt{J+1}) \sqrt{\frac{1}{2S}\ln  \left(\frac{2}{\delta} \right)}
        \end{align}
    Using that $1\leq \Lambda \sqrt{J+1}$, the expression simplifies to,
    \begin{align}
            \mathbb{E}_{W\sim D(U) }&\left[s_U(W)\left(f(W)-\boldsymbol{\alpha}^T \boldsymbol{\phi}_{\infty}(W) \right)\right] \leq \frac{1}{S}\sum_{i=1}^S  s_U(W_i)\left(f(W_i)-\boldsymbol{\alpha}^T \boldsymbol{\phi}_{\infty}(W_i)  \right)\nonumber \\ &+ \frac{\Lambda \left \| O \right \|_2 \sqrt{J+1}}{\sqrt{S}} \left( 2+\sqrt{72 \log \left( \frac{2}{\delta} \right)}  \right) 
    \end{align}
    with at least probability $1-\delta$. Next, to remove the constraint $\| \boldsymbol{\alpha} \|_2 \leq \Lambda$, we adopt the same procedure as in \cite{huang2021power}. That is, we substitute $\delta= \delta'/(2\Lambda^2)$ and apply the union bound for $\Lambda \in \{1,2,3,\cdots\}$, yielding
        \begin{align}
            \mathbb{E}_{W\sim D(U) }&\left[s_U(W)\left(f(W)-\boldsymbol{\alpha}^T \boldsymbol{\phi}_{\infty}(W) \right)\right] \leq \frac{1}{S}\sum_{i=1}^S  s_U(W_i)\left(f(W_i)-\boldsymbol{\alpha}^T \boldsymbol{\phi}_{\infty}(W_i)  \right)\nonumber \\ &+ \frac{\left \lceil \| \boldsymbol{\alpha} \|_2 \right \rceil \left \| O \right \|_2 \sqrt{(J+1)}}{\sqrt{S}} \left( 2+\sqrt{72 \log \frac{4 \left \lceil \| \boldsymbol{\alpha} \|_2 \right \rceil^2}{\delta'}}  \right) 
        \end{align}
    with at least probability $1-\delta'$. 

    Given that the same inequality holds for the function $-s_U(W)\left(f(W)-\boldsymbol{\alpha}^T \boldsymbol{\phi}_{\infty}(W)\right)$, it follows that
    \begin{align}
            \Big|\mathbb{E}_{W\sim D(U) }&\left[s_U(W)\left(f(W)-\boldsymbol{\alpha}^T \boldsymbol{\phi}_{\infty}(W) \right)\right] - \frac{1}{S}\sum_{i=1}^S  s_U(W_i)\left(f(W_i)-\boldsymbol{\alpha}^T \boldsymbol{\phi}_{\infty}(W_i)  \right) \Big | \nonumber  \\ &\leq \frac{\left \lceil \| \boldsymbol{\alpha} \|_2 \right \rceil \left \| O \right \|_2 \sqrt{(J+1)}}{\sqrt{S}} \left( 2+\sqrt{72 \log \frac{8 \left \lceil \| \boldsymbol{\alpha} \|_2 \right \rceil^2}{\delta}}  \right) 
    \end{align}
    with at least probability $1-\delta$. Therefore,
    \begin{align}
        &\Big|\mathbb{E}_{W\sim D(U) }\left[s_U(W)\left(f(W)-\boldsymbol{\alpha}^T \boldsymbol{\phi}_{\infty}(W) \right)\right]  \Big |  \leq  \Big|\frac{1}{S}\sum_{i=1}^S  s_U(W_i)\left(f(W_i)-\boldsymbol{\alpha}^T \boldsymbol{\phi}_{\infty}(W_i)  \right) \Big | \nonumber \\ &  + \Big|\mathbb{E}_{W\sim D(U) }\left[s_U(W)\left(f(W)-\boldsymbol{\alpha}^T \boldsymbol{\phi}_{\infty}(W) \right)\right] - \frac{1}{S}\sum_{i=1}^S  s_U(W_i)\left(f(W_i)-\boldsymbol{\alpha}^T \boldsymbol{\phi}_{\infty}(W_i)  \right) \Big | \nonumber \\ & \leq \Big|\frac{1}{S}\sum_{i=1}^S  s_U(W_i)\left(f(W_i)-\boldsymbol{\alpha}^T \boldsymbol{\phi}_{\infty}(W_i)  \right) \Big |+ \frac{\left \lceil \| \boldsymbol{\alpha} \|_2 \right \rceil \left \| O \right \|_2 \sqrt{(J+1)}}{\sqrt{S}} \left( 2+\sqrt{72 \log \frac{8 \left \lceil \| \boldsymbol{\alpha} \|_2 \right \rceil^2}{\delta}}  \right)
    \end{align}
    with at least probability $1-\delta$.

    \end{proof}

\setcounter{equation}{0}
    \section{Numerical Experiments testing the Generalization Error}\label{numerical_experiments_testing}

    In this appendix, we conduct numerical experiments to examine whether the difference
    \begin{align}
    \Delta(U) &:= \left| \frac{N(\boldsymbol{\theta})}{S}\sum_{i=1}^S s_U(W_i) \left( f(W_i)-\boldsymbol{\alpha}^T \boldsymbol{\phi}_{\infty}(W_i)\right)\right|-|f(U)-\boldsymbol{\alpha}^T \boldsymbol{\phi}_{\infty}(U)|
    \end{align}
    scales as predicted by the generalization bound from Lemma \ref{lemma_2}. This analysis is performed because generalization bounds are often not particularly tight, and therefore, better scaling may be achievable in general.
    
    In particular, Figures \ref{figure_Delta_1} and \ref{figure_Delta_2} depict estimates of $\mathbb{E}[\Delta(U)]$ for different values of $S$, where the mean is estimated over $10^3$ random circuits $U$. Along with the estimate of the expected value, we also plot estimates of the standard deviation of $\Delta(U)$, denoted as $\sigma(\Delta(U))$. This variable should not be confused with the variance of the estimate of $\mathbb{E}[\Delta(U)]$, which is $\sigma(\Delta(U))^2/S$. These 3-qubit random circuits consist of $\ell=25$ gates, each chosen uniformly at random from the set $\{\text{CNOT}, R_X, R_Y, R_Z\}$. To ensure a certain level of noise, since in the simulations we do not introduce noise for the 1-qubit gates, if a random circuit contains fewer than 6 CNOT gates, rotations are substituted with CNOT gates until the count reaches 6. In both figures, rotations are substituted by random gates from the set $\{I, X, Y, Z, \sqrt{X}, \sqrt{Y}, \sqrt{Z}\}$. The probability of each gate is proportional to the absolute value of the corresponding coefficient given in the decompositions \eqref{decomposition_X}, \eqref{decomposition_Y}, and \eqref{decomposition_Z}. In Figure \ref{figure_Delta_1}, all rotations are substituted, whereas in Figure \ref{figure_Delta_2}, 6 randomly selected rotations are not substituted, causing the value of $N(\boldsymbol{\theta})$ to decrease. This enables a comparison of the effect of this term. The value of $\boldsymbol{\alpha}$ is determined using a training set of fixed size 100 for all simulations. We adopt this approach instead of training with the random samples obtained for two reasons: $(i)$ $\boldsymbol{\alpha}$ remains fixed in this case and does not vary with $S$, facilitating the observation of $\Delta$'s dependency on this parameter, and $(ii)$ for computational efficiency. Finally, outliers have been removed in both cases, defined as values of $\Delta(U)$ where $|\Delta(U)| > 7000$ for Figure \ref{figure_Delta_1} and $|\Delta(U)| > 60$ for Figure 1b. These outliers constitute approximately 0.5\% of the data for Figure \ref{figure_Delta_1} and 0.2\% for Figure \ref{figure_Delta_2}.

        \begin{figure*}[htb]

        \begin{subfigure}[b]{0.5\textwidth}
          \centering
          \centerline{\includegraphics[width=6.8cm]{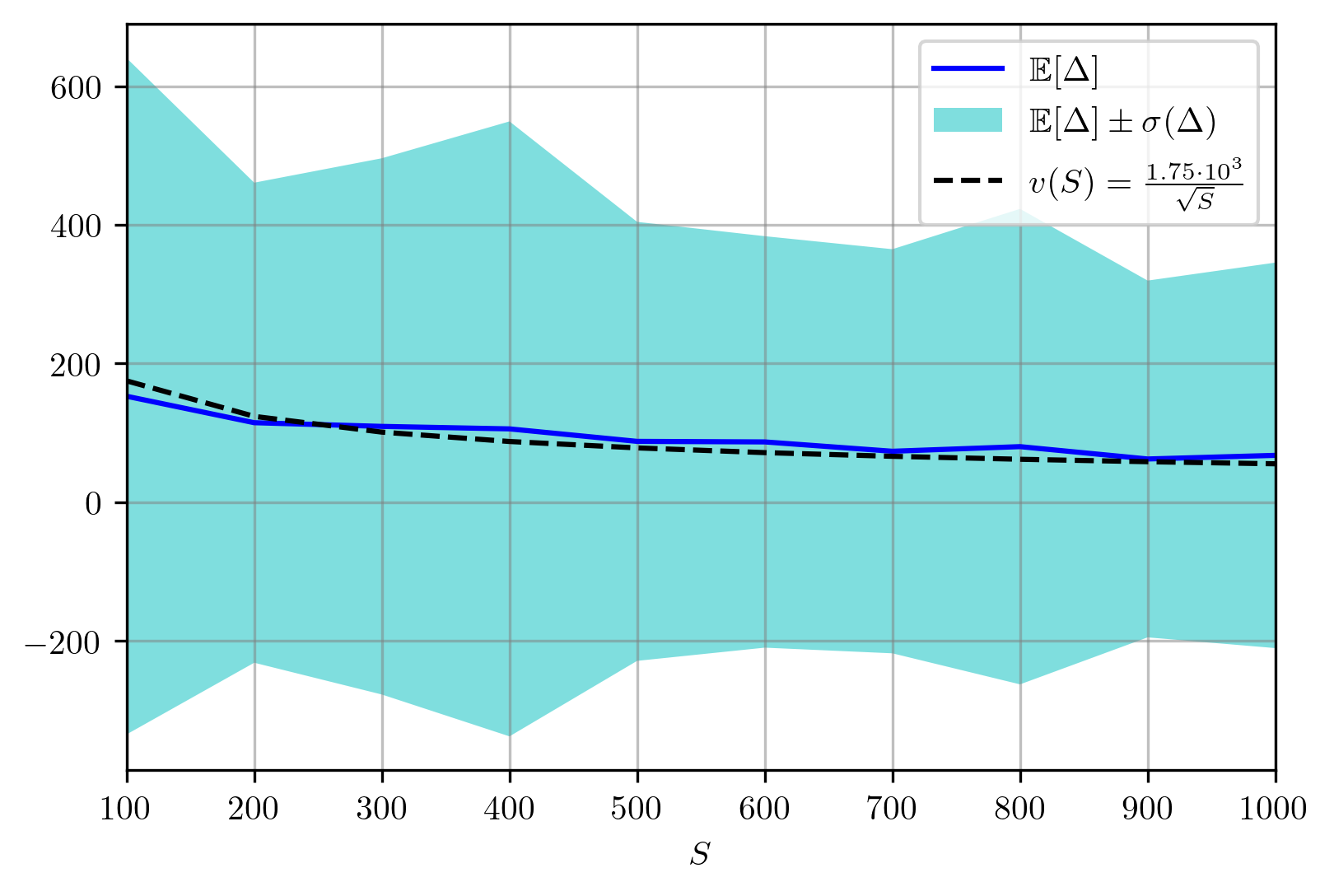}}
          \caption{}
          \label{figure_Delta_1}\medskip
        \end{subfigure}
        \hfill
        \begin{subfigure}[b]{0.5\textwidth}
          \centering
          \centerline{\includegraphics[width=6.8cm]{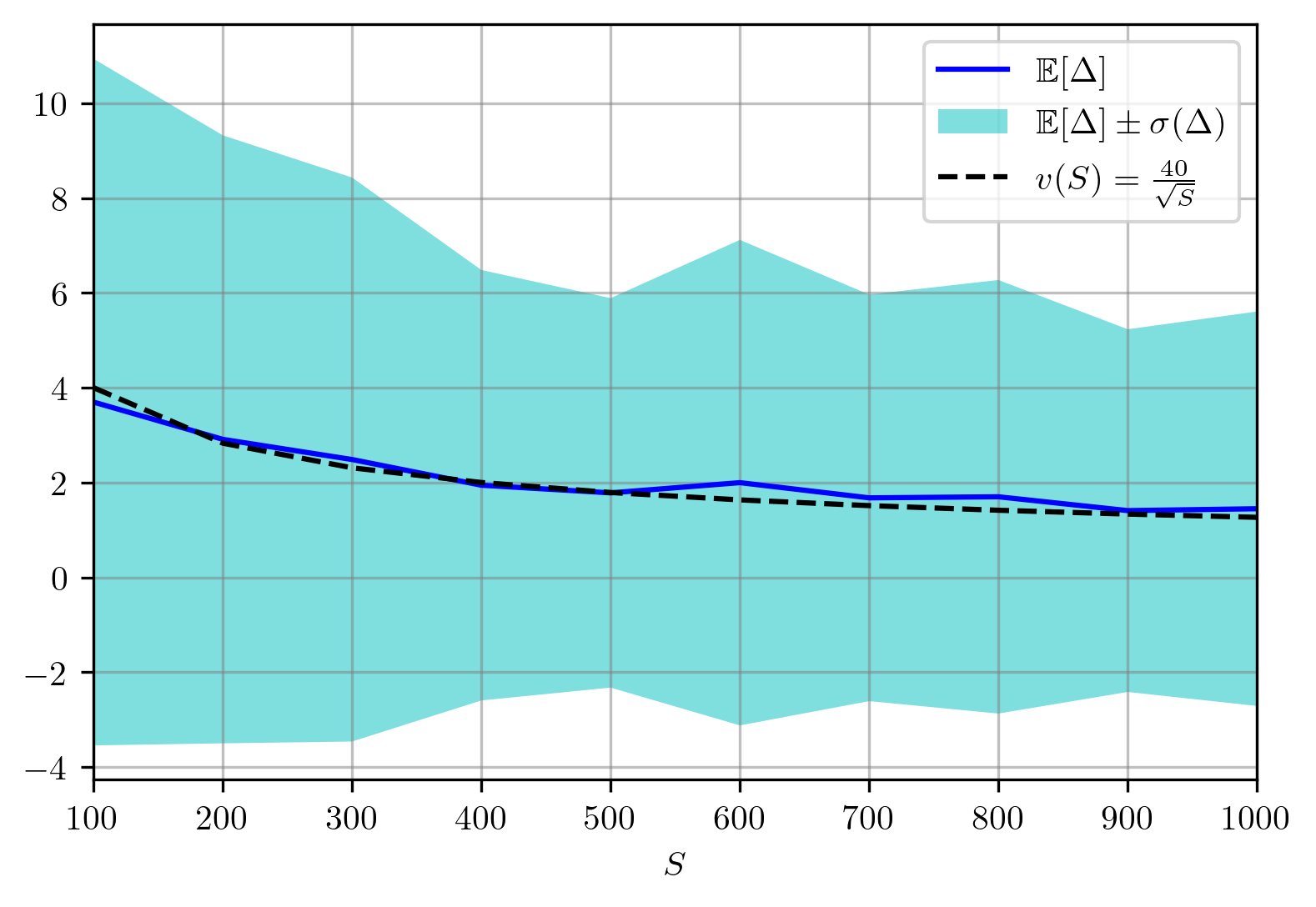}}
          \caption{}
          \label{figure_Delta_2} \medskip
        \end{subfigure}

        \caption{For Figure (a), $\Delta(U)$ is depicted for different values of $S$, with all non-Clifford gates substituted. In Figure (b), $\Delta(U)$ is shown for the same range of $S$, but with 6 non-Clifford gates left unsubstituted.}
        \label{figure_2_A}
        
        \end{figure*}

   In both figures, we observe that $\Delta(U)$ decreases with $S$ approximately according to $1/\sqrt{S}$, as determined by the bound. Additionally, since the only difference between Figure \ref{figure_Delta_1} and Figure \ref{figure_Delta_2} is the presence of those $6$ non-substituted rotations, we can observe how the normalization factor $N(\boldsymbol{\theta})$ affects $\Delta(U)$. Specifically, the average normalization constant in Figure \ref{figure_Delta_1} is significantly larger than in Figure \ref{figure_Delta_2}, as $N(\boldsymbol{\theta})$ grows exponentially with the number of replaced rotations. Formally,
    \begin{equation}
        \left(3\left(\frac{1}{4}+\frac{1}{\pi}\right)\right)^{\ell_r}\leq \mathbb{E}_{\boldsymbol{\theta}} [N(\boldsymbol{\theta})] \leq \left(\frac{6}{\pi}\right)^{\ell_r}
    \end{equation}
    where the expectation value is over the uniform distribution in $[0,2\pi]^{\ell_r}$, and $\ell_r$ denotes the number of rotations substituted. The latter follows from the fact that each weight $w(\boldsymbol{\theta})$ can be expressed as a product of $\ell_r$ terms, where each term is of the form $(1 + \cos \theta_i - \sin \theta_i)/{2}$, $(1 - \cos \theta_i - \sin \theta_i)/{2}$, or $\sin \theta_i$. Therefore, using that the number of weights is $3^{\ell_r}$, and 
    \begin{equation*}
        \mathbb{E}_\theta \left[ \left|\frac{1+\cos \theta -\sin \theta}{2}\right|\right]= \frac{1}{4}+\frac{1}{\pi}, \hspace{1cm}       \mathbb{E}_\theta \left[ \left|\frac{1-\cos \theta -\sin \theta}{2}\right|\right]=\frac{1}{4}+\frac{1}{\pi}
    \end{equation*}
    \begin{equation}
        \mathbb{E}_\theta \left[ \left| \sin \theta \right|\right]= \frac{2}{\pi}
    \end{equation}
    we obtain the desired result.

    Lastly, we study how changes in the feature map dimension, $J+1$, affect $\Delta(U)$ for the insertion method. We find that $\Delta(U)$ does not exhibit significant dependency on $J$. Similarly, the product $\| \boldsymbol{\alpha} \|_2 \sqrt{J+1}$ remains constant across different values of $J$. Consequently, both $\Delta(U)$ and $\| \boldsymbol{\alpha}\|_2 \sqrt{J+1}$ follow a similar behavior, consistent with what is predicted by the generalization bound.  

        \begin{figure*}[htb]

        \begin{subfigure}[b]{0.5\textwidth}
          \centering
          \centerline{\includegraphics[width=6.8cm]{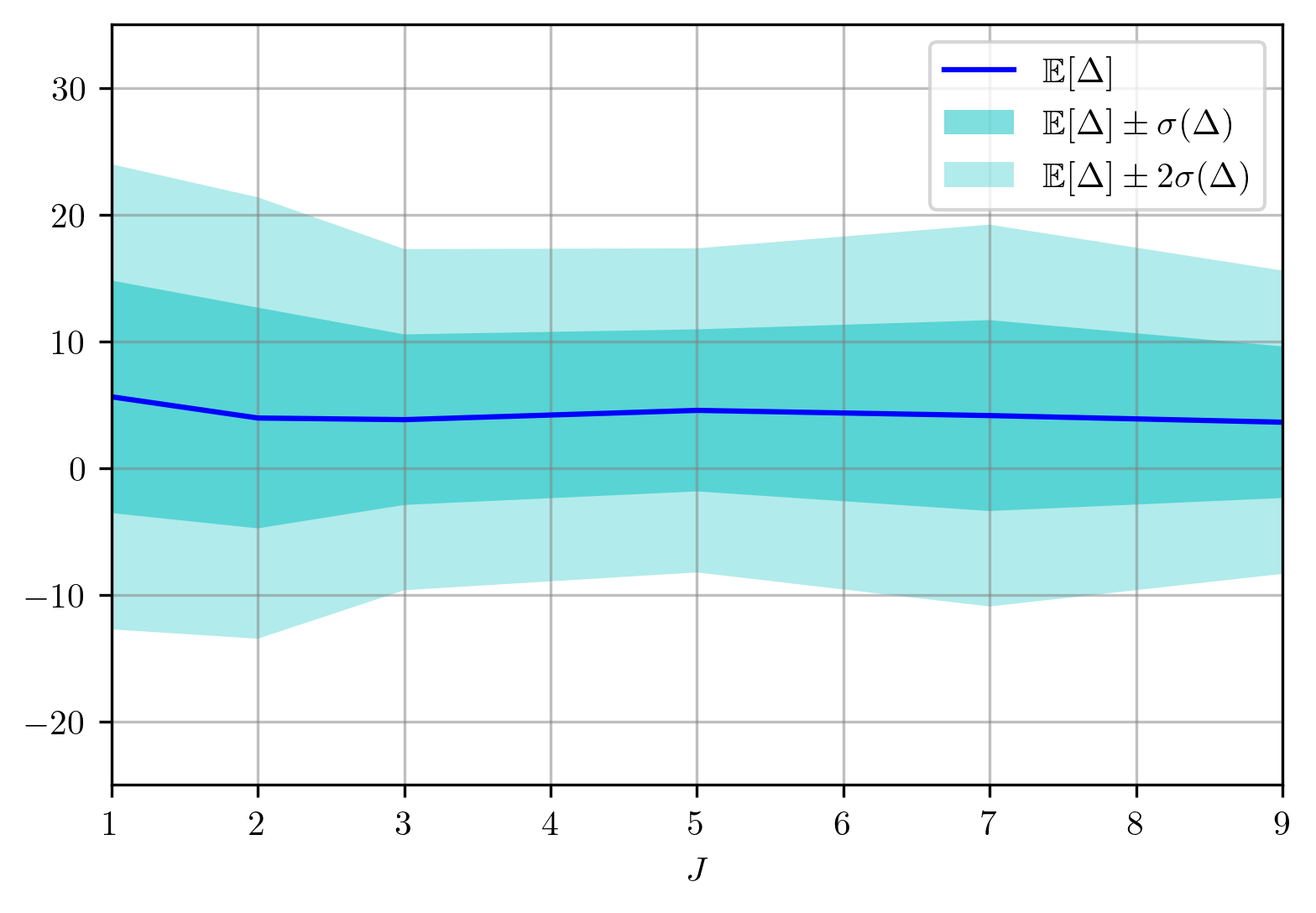}}
          \caption{}
          \label{figure_Delta_3}\medskip
        \end{subfigure}
        \hspace{0.25cm}
        \begin{subfigure}[b]{0.5\textwidth}
          \centering
          \centerline{\includegraphics[width=6.8cm]{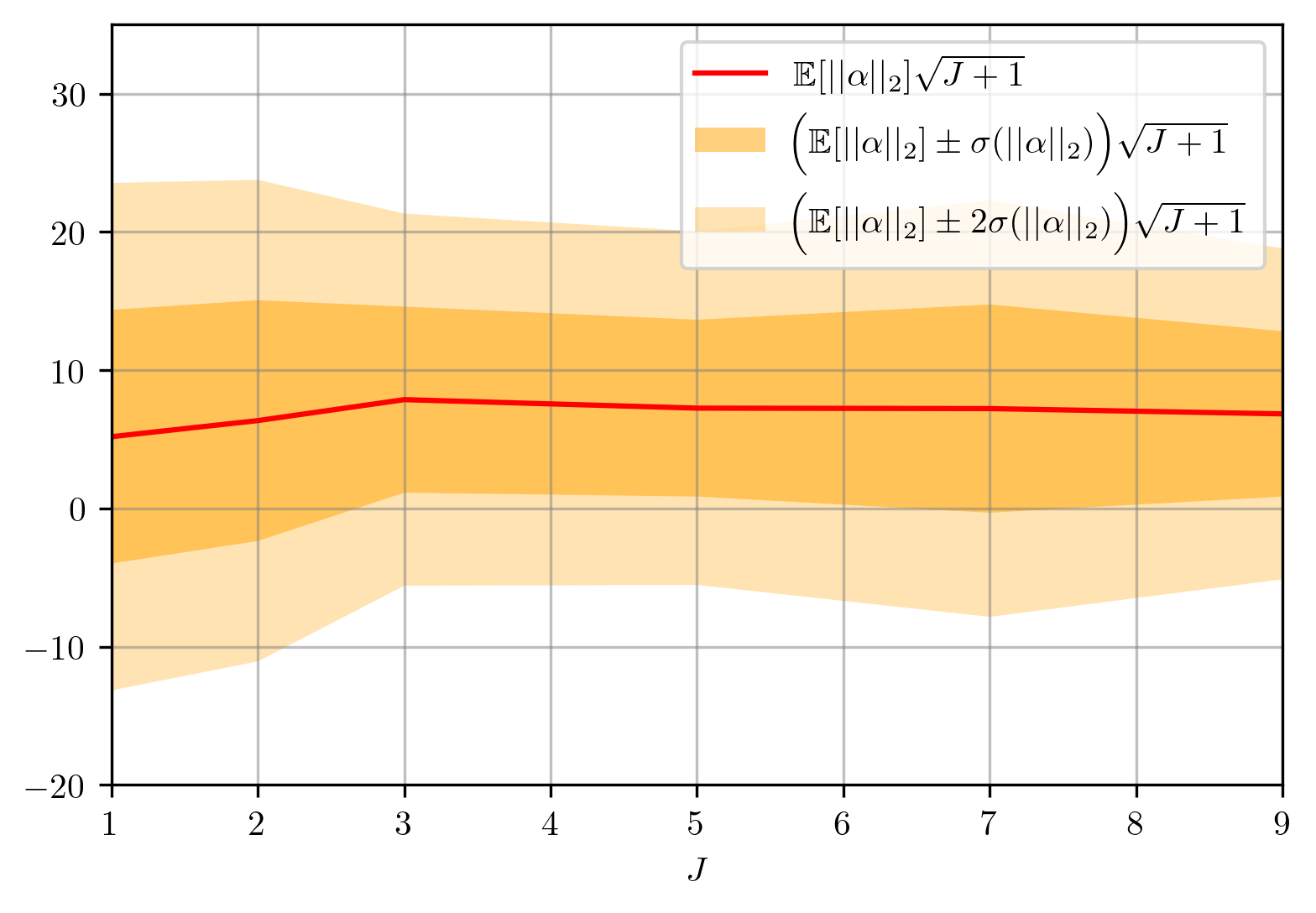}}
          \caption{}
           \label{figure_Delta_4}\medskip
        \end{subfigure}

        \caption{In Figure (a) the expected value of $\Delta(U)$ is displayed for different values of $J$. Figure (b) presents a plot of $\mathbb{E}[\|\boldsymbol{\alpha}\|_2 \sqrt{J+1}]$ across various values of $J$.}
        \end{figure*}

    Figures \ref{figure_Delta_3} and \ref{figure_Delta_4} are obtained through simulation of $10^3$ random circuits, using the same parameters as in Figure \ref{figure_2_A}. The main difference from the previous figures is that the training is performed with the random circuits used to compute $\Delta(U)$. In this case, we choose $S=100$.

\setcounter{equation}{0}
    \section{Proof of Theorem \ref{lower_bound}}\label{last_appendix}

    In this appendix, the proof of Theorem \ref{lower_bound} is shown. For completeness, the Theorem \ref{lower_bound} is restate below. 
    \begin{theorem**}\label{lower_bound_app}
        For any set of unitaries $\{U_m\}_{m=1}^M$, such that
        \begin{equation}
            \left|f(U_i)-f(U_j) \right|\geq 2\epsilon 
        \end{equation}
        for all $i\neq j$, and for any estimator, $\hat{f}(U)$, that depends on $U$ through $\boldsymbol{\phi}(U)$ and $ \mathcal{T}_S(U)$, where
        \begin{equation}
            \boldsymbol{\phi}(U)=[ 1 ,\,  \phi(\tilde{\mathcal{P}}_1 (U)), \phi(\tilde{\mathcal{P}}_2 (U)), \cdots, \phi(\tilde{\mathcal{P}}_J (U))]^T
        \end{equation}
       and $\mathcal{T}_S(U)=\{\boldsymbol{\phi}(W_i), f(W_i)\}_{i=1}^S$,

        \begin{enumerate}[(i)]
            \item if we assume that the noise of each circuit $\tilde{\mathcal{P}_j}(U_m)$ can be modeled as a depolarizing channel with parameter $p$, i.e., $\tilde{\mathcal{P}_j}(U_m) (\rho )= (1-p){\mathcal{P}_j}(U_m)( \rho) + p \frac{I}{2^n}$, and that the same holds for the circuits in the training set, i.e., $\tilde{\mathcal{P}}_j(W) (\rho )= (1-p){\mathcal{P}}_j(W)( \rho) + p \frac{I}{2^n}$, then :

            The variables $J,S$, and $N$ satisfy 
            
                \begin{equation}
                    \frac{(1-P_{\epsilon}-\beta(S)) \log M -H(P_{\epsilon})}{n\,\left(1- p\right)} \leq N J (S+1)
                \end{equation}

            where $H(\cdot)$ is the binary entropy function, $P_{\epsilon}$ denotes the average probability that the estimation error exceeds $\epsilon$, i.e., $\frac{1}{M} \sum_{i=1}^M \mathbb{P} (|f(U_i)-\hat{f}(U_i)|> \epsilon )$, and $\beta(S)<1$ depends on the methodology used to generate the training set.

            \item Similarly, if all circuits $\mathcal{P}_j(U_m)$ have depth $d$, and we assume that each noisy layer of gates $\tilde{\mathcal{G}}$ can be modeled as $D_p^{\otimes n}\circ \mathcal{G}$, where $D_p(\rho)$ denotes a 1-qubit depolarizing channel with parameter $p$, then
            \begin{equation}
                \frac{(1-P_{\epsilon}-\beta(S)) \log M -H(P_{\epsilon})}{n\,\left(1- p\right)^{2d}} \leq N J (S+1)
            \end{equation}

        \end{enumerate}
    \end{theorem**}

    \begin{proof}

        Let r.v. $X(j,U)$ denote the outcome of measuring the noisy circuit $\tilde{\mathcal{P}}_j (U)$. The vector $[X(1,U),\cdots,X(J,U)]^T$ is denoted by $\boldsymbol{X}(U)$. A sequence of samples of this random variable, $\boldsymbol{x}_1(U),\cdots, \boldsymbol{x}_N(U)$, is represented as $\boldsymbol{x}(U)^N$. The r.v. $G$ is uniformly distributed over $\{1,\cdots, M\}$. The \textit{estimation} of $G$ is the r.v. $\hat{G}$, defined as
        \begin{equation}
            \hat{G}:= \argmin_{m'\in [1:M] }|f(U_{m'})-h(\boldsymbol{\phi}(U_G), \mathcal{T}_S(U_G))|    
        \end{equation}
        where $\hat{f}(U)=h(\boldsymbol{\phi}(U), \mathcal{T}_S(U))$.
        
        The mutual information between $G$ and $\hat{G}$ can be lower bounded using Fano's inequality,
        \begin{align}\label{lower_bound_mutual_information}
            I(G;\hat{G})&= H(G)-H(G|\hat{G})\geq \log M -H(\mathbb{P}(G\neq \hat{G}))-\mathbb{P}(G\neq \hat{G}) \log M  \nonumber \\ &\geq \log M -H(P_{\epsilon})-P_{\epsilon} \log M 
        \end{align}
        where $P_{\epsilon}=\frac{1}{M} \sum_{i=1}^M \mathbb{P} (|f(U_i)-\hat{f}(U_i)|> \epsilon )$. Next, the mutual information can be upper bounded as follows
        \begin{align}\label{upper_bound_mutual_information}
            I(G;\hat{G})&\leq I(G; \boldsymbol{\phi}(U_G),\mathcal{T}_S(U_G))\leq I(G; \boldsymbol{X}(U_G)^N,\mathcal{T}_S(U_G))\nonumber \\ &= I(G;\mathcal{T}_S(U_G))+I(G;\boldsymbol{X}(U_G)^N |\mathcal{T}_S(U_G))
            \nonumber \\ &\leq  I(G;\mathcal{T}_S(U_G))+N\cdot I(G;\boldsymbol{X}(U_G) |\mathcal{T}_S(U_G))
        \end{align}
        where the first and second inequalities follow from the data processing inequality, and the third inequality uses that 
\begin{align}\label{D_46}
    I(G;\boldsymbol{X}(U_G)^N |\mathcal{T}_S(U_G))&=\sum_{i=1}^N I(G;\boldsymbol{X}_i(U_G)|\mathcal{T}_{S}(U_G), \boldsymbol{X}(U_G)^{i-1} ) \nonumber \\ & = \sum_{i=1}^N H( \boldsymbol{X}_i(U_G) |\mathcal{T}_{S}(U_G), \boldsymbol{X}(U_G)^{i-1})-H(\boldsymbol{X}_i(U_G)|G,\mathcal{T}_{S}(U_G))
    \nonumber \\ &\leq  N \cdot I(G;\boldsymbol{X}(U_G)|\mathcal{T}_S(U_G))
\end{align}
where the first equality follows from applying the chain rule for mutual information. The second equality uses that $H( \boldsymbol{X}_i(U_G) |G,\mathcal{T}_{S}(U_G), \boldsymbol{X}(U_G)^{i-1})=H( \boldsymbol{X}_i(U_G) |G,\mathcal{T}_{S}(U_G))$, and the last inequality follows from $H(A|B,C)\leq H(A|B)$.

Therefore, combining \eqref{lower_bound_mutual_information} and \eqref{upper_bound_mutual_information},
\begin{equation}
    (1-P_{\epsilon})\log M -H(P_{\epsilon}) \leq I(G;\mathcal{T}_S(U_G))+N\cdot I(G;\boldsymbol{X}(U_G) |\mathcal{T}_S(U_G))
\end{equation}
Next, we define ensemble $\mathcal{E}:=\{1/M,\tilde{\sigma}_{m}\}$, where 
\begin{equation}
   \tilde{\sigma}_m:= \left(\bigotimes_{j=1}^J \tilde{\mathcal{P}}_j\left(\mathcal{U}_m\right) \right)(\rho_0^{\otimes J})
\end{equation}
and $\rho_0=\ket{0} \bra{0}^{\otimes\, n}$. Using this definition, we can upper bound the coefficient associated with $N$ as follows,  
\begin{align}
    I(G;\boldsymbol{X} |\mathcal{T}_S(U_G)) &= H(\boldsymbol{X}|\mathcal{T}_S(U_G))-H(\boldsymbol{X}|G,\mathcal{T}_S(U_G)) \nonumber \\ &\leq  H(\boldsymbol{X})-H(\boldsymbol{X}|G)= I(G;\boldsymbol{X}) \nonumber \\ &\leq  I_{\mathrm{acc}}(\mathcal{E}) \leq \chi (\mathcal{E})
\end{align}
where the first inequality follows from $H(\boldsymbol{X}|\mathcal{T}_S(U_M))\leq H(\boldsymbol{X})$, the second inequality follows from the definition of accessible information, and the third inequality uses Holevo's theorem. Furthermore, if we assume that the noise is modelled as a depolarizing channel, then 
\begin{align}\label{upper_bound_Holevo}
    \chi (\mathcal{E})&= H(\tilde{\sigma})- \frac{1}{M} \sum_{m=1}^{M} H(\tilde{\sigma}_m)\nonumber \\ &\leq nJ  - \frac{1}{M} \sum_{m=1}^{M} \sum_{j=1}^J H\left( \tilde{\mathcal{P}}_j(\mathcal{U}_m)(\rho_0)\right) \nonumber \\ & =  nJ  - \frac{1}{M} \sum_{m=1}^{M} \sum_{j=1}^J H\left( \left(1-p \right) \cdot\mathcal{P}_j(\mathcal{U}_m)(\rho_0)+ p \frac{I}{2^n}\right) \nonumber \\ & \leq 
    nJ  - \frac{1}{M} \sum_{m=1}^{M} \sum_{j=1}^J \left(1-p \right) H\left( \mathcal{P}_j(\mathcal{U}_m)(\rho_0)\right)+ p H\left(\frac{I}{2^n}\right) \nonumber \\ &= n J \,(1-p) 
\end{align}
where the first inequality uses that $H(\sigma)\leq \log d$, and $H(\rho\otimes \sigma)=H(\rho)+H(\sigma)$, the second inequality follows from $H(\alpha \rho+(1-\alpha)\sigma)\geq \alpha H(\rho)+ (1-\alpha) H(\sigma)$. The last step uses that $\mathcal{P}_j(\mathcal{U}_m)(\rho_0)$ is a pure state, and that $H(I/2^n)=n$

Next, we bound the factor associated with the training set, i.e., $I(G;\mathcal{T}_S(U_G))$. To do this we use the fact that each random circuit $W$ in the training set is generated by a deterministic function $T$, with inputs $U_G$ and a random string $Z$ that is completely independent of $G$. In other words, $W_i=T(U_G,z_i)$.
\begin{align}
    I(G;\mathcal{T}_S(U_G))&=I(G; \boldsymbol{f}, \Phi )=I(G; \boldsymbol{f})+I(G;\Phi |\boldsymbol{f}) \nonumber \\&=I(G; \boldsymbol{f})+H(\Phi|\boldsymbol{f})-H(\Phi|G,\boldsymbol{f})  \nonumber \\& \leq I(G; \boldsymbol{f})+H(\Phi)-H(\Phi|G,\boldsymbol{f},Z^S) \nonumber \\& = I(G; \boldsymbol{f})+I(G; \Phi,Z^S)=I(G; \boldsymbol{f})+I(G; \Phi |Z^S)
\end{align}
where $\boldsymbol{f}=[f(W_1),f(W_2),\cdots, f(W_S)]^T$, and $\Phi=[ \boldsymbol{\phi}(W_1), \boldsymbol{\phi}(W_2),\cdots,  \boldsymbol{\phi}(W_S)]^T$. For the first inequality, we utilize that $H(A|B)\leq H(A)$, and the last two equalities use $H(\Phi|G,\boldsymbol{f},Z^S)=H(\Phi|G,Z^S)$, and that $I(G;Z^S)=0$.

To bound $I(G; \Phi |Z^S)$, we proceed as follows:
\begin{align}\label{upper_2}
    I(G; \Phi |Z^S)&=I(G; \boldsymbol{\phi}(W_1), \cdots, \boldsymbol{\phi}(W_S)| Z^S )  \nonumber \\ &= I(G; \boldsymbol{\phi}(W_1)|Z^S)+\sum_{i=2}^S I(G; \boldsymbol{\phi}(W_i) |Z^S, \underbrace{\boldsymbol{\phi}(W_1),\cdots, \boldsymbol{\phi}(W_{i-1})}_{\Phi^{i-1}}) \nonumber \\ &= I(G; \boldsymbol{\phi}(W_1)|Z_1)+\sum_{i=2}^S H(\boldsymbol{\phi}(W_i)|Z^S, \Phi^{i-1}) -H(\boldsymbol{\phi}(W_i)|Z^S, \Phi^{i-1},G)  \nonumber \\  & \leq  I(G;\boldsymbol{\phi}(W_1)|Z_1)+      \sum_{i=2}^S H(\boldsymbol{\phi}(W_i)|Z_i) -H(\boldsymbol{\phi}(W_i)|Z_i,G)  \nonumber \\ & = S \cdot I(G; \boldsymbol{\phi}(W_1)|Z_1) \nonumber \\ & \leq   S \, N \cdot I(G;\boldsymbol{X}(W_1)|Z_1)=   S \, N \cdot I(G;\boldsymbol{X}(T(U_G,Z_1))|Z_1) \nonumber \\ & \leq S \, N \cdot \mathbb{E}_{z_1} \left[ I_{\mathrm{acc}}\left( \mathcal{E}(z_1)\right)\right]\leq S \, N \cdot \mathbb{E}_{z_1} \left[\chi \left( \mathcal{E}(z_1)\right)\right]\leq n S N J \,(1-p)
\end{align}
where the ensemble $\mathcal{E}(z)$ is defined as
\begin{equation}
   \mathcal{E}(z)=\left\{\frac{1}{M}, \left(\bigotimes_{j=1}^J \tilde{\mathcal{P}}_j\left(T(U_m,z)\right) \right)(\rho_0^{\otimes J})\right \}
\end{equation}
The second inequality is analogous to \eqref{D_46}, the third and fourth inequalities follow from the definition of the accessible information, and Holevo's theorem, respectively. The last inequality uses \eqref{upper_bound_Holevo}.

Finally, defining, $\beta(S):=I(G;f(W_1),\cdots, f(W_S))/\log M$, which depends on the algorithm used to generate the training set, and the previous upper bounds, the following inequality follows
\begin{equation}\label{upper_final_}
    I(G;\hat{G})\leq \beta(S) \log M + n (S+1) N J (1-p)
\end{equation}
which combined with \eqref{lower_bound_mutual_information} implies,
\begin{equation}
    \frac{(1-P_{\epsilon}-\beta(S)) \log M-H(P_{\epsilon})}{n\, (1-p)} \leq N J(S+1)
\end{equation}
To prove the second point, we can use the same methodology, with the additional result,
\begin{equation}
    H\left( \tilde{\mathcal{P}}_j(\mathcal{U}_m)(\rho_0)\right)\geq  n\left(1-\left(1-p \right)^{2d}\right)
\end{equation}
for all $j \in [1:J]$. To do this, we apply the strong data processing inequality \cite{hirche2022contraction,beigi2020quantum}:
\begin{equation}
    D\left(D_p^{\otimes n} (\sigma)\,||\, \frac{I}{2^n} \right)\leq (1-p)^2  D\left(\sigma\,|| \,\frac{I}{2^n}\right)
\end{equation}
where $D(\rho||\sigma)$ denotes the quantum relative entropy. Next, using that 
\begin{align}
    D\left(\sigma\,|| \,\frac{I}{2^n}\right)& =-H(\sigma)-\Tr(\sigma \log \frac{I}{2^n})\nonumber \\&=-H(\sigma) +n
\end{align}
it follows that,
\begin{equation}
    H\left(D_p^{\otimes n} (\sigma) \right)\geq n(1-(1-p)^2)+(1-p)^2 H(\sigma)
\end{equation}
Note that any circuit $\mathcal{P}_j(U_m)$ has depth $d$ by assumption. That is, any noisy circuit $\tilde{\mathcal{P}}_j(U_m)$ can be expressed as $D_p^{\otimes n}\circ \mathcal{G}_d \circ \cdots \circ D_p^{\otimes n}\circ \mathcal{G}_1$, where $\mathcal{G}_i(\rho)=G_i \rho G_i^{\dagger}$ denotes the $i^{th}$ layer of the circuit. Using this expression of $\tilde{\mathcal{P}}_j(U_m)$, and the fact that $H(U\rho U^{\dagger})=H(\rho)$ for any unitary $U$, we conclude that
\begin{equation}
    H\left(D_p^{\otimes n} (\sigma_{j}) \right)\geq n(1-(1-p)^2)+(1-p)^2 H(D_p^{\otimes n} (\sigma_{j-1}))
\end{equation}
where $\sigma_{j}= \mathcal{G}_j \circ \cdots \circ D_p^{\otimes n}\circ \mathcal{G}_1(\sigma_0)$, and $\sigma_0=\ket{0}\bra{0}^{\otimes n}$. Therefore, solving this recursion, we obtain the desired result,
\begin{equation}
    H\left(\tilde{\mathcal{P}}_j(U_m)(\rho_0) \right)=H\left(D_p^{\otimes n} (\sigma_{d}) \right)\geq n(1-(1-p)^{2d})
\end{equation}

    \end{proof}

\end{appendices}
\end{document}